\documentclass{LMCS}

\usepackage{enumerate}
\usepackage{hyperref}

\theoremstyle{plain}

\usepackage{amsmath}
\usepackage{amsthm}
\usepackage{amssymb}
\usepackage{stmaryrd}
\usepackage{xy}
\xyoption{all}

\newdir{|>}{!/4.5pt/\dir{|}
  *:(1,-.2)\dir^{>}
  *:(1,+.2)\dir_{>}}
\newdir{ >}{{}*!/-10pt/\dir{>}}
\newdir{ (}{{}*!/-10pt/\dir^{(}}

\newdir{pb}{:(1,-1)@^{|-}}
\def\pb#1{\save[]+<24 pt,0 pt>:a(#1)\ar@{pb{}}[]\restore}
\newdir{pbb}{:(1,-1)@^{|-}}
\def\pbb#1{\save[]+<18 pt,0 pt>:a(#1)\ar@{pbb{}}[]\restore}
\newdir{pbbb}{:(1,-1)@^{|-}}
\def\pbbb#1{\save[]+<36 pt,0 pt>:a(#1)\ar@{pbbb{}}[]\restore}

\newenvironment{proofof}[1]{\trivlist\item[]{\em Proof of #1\/}:}%
{\unskip\nobreak\hskip 2em plus 1fil\nobreak$\Box$
\parfillskip=0pt%
\endtrivlist}

\newcommand{\hide}[1]{}
\newcommand{\secref}[1]{Section~\ref{#1}}
\newcommand{\lemmaref}[1]{Lemma~\ref{#1}}
\newcommand{\factref}[1]{Fact~\ref{#1}}
\newcommand{\propref}[1]{Propn.~\ref{#1}}

\newcommand{\thmref}[1]{Theorem~\ref{#1}}

\newcommand{\dgmref}[1]{diagram~\eqref{#1}}
\newcommand{\itemref}[1]{item~(\ref{#1})}
\newcommand{\Itemref}[1]{Item~(\ref{#1})}
\newcommand{\ie}{{i.e.}}
\newcommand{\etc}{{etc.}}
\newcommand{\viz}{{viz.}}
\newcommand{\textcitesec}{Sec.}

\newcommand{\mathprefix}[2]{$#2$\nobreakdash-\mbox{}#1}
\newcommand{\picalculus}{\mathprefix{calculus}\pi}
\newcommand{\gcoalgtext}{\mathprefix{coalgebra}}
\newcommand{\AMbisimtext}{AM-bisimulation}
\newcommand{\AMprecong}{AM-pre\-con\-gru\-ence}
\newcommand{\HJbisimtext}{HJ-bisimulation}

\newcommand{\Set}{{\mathbf{Set}}}
\newcommand{\Icat}{{\mathbb{I}}}

\newcommand{\Fnecat}{{\mathbb{F}^+}}
\newcommand{\Ccat}{{\mathbb{C}}}
\newcommand{\objC}{C}

\newcommand{\PshfI}{[\Icat,\Set]}

\newcommand{\PshfFne}{[\Fnecat,\Set]}
\newcommand{\PshfC}{[\Ccat,\Set]}
\newcommand{\ShfI}{\mathrm{Sh}_{\neg\neg}(\Icat)}

\newcommand{\inv}{^{\mathrm{-1}}}

\newcommand{\SetPow}{\mathcal{P}}

\newcommand{\DistF}{\mathcal{D}_\mathrm{f}}
\newcommand{\SetPowF}{\mathcal{P}_\mathrm{f}}
\newcommand{\SetPowPF}{\mathcal{P}_{{\mathrm{pf}}}}

\newcommand{\CCSLabels}{L}
\newcommand{\CCSLab}{l}
\newcommand{\setofcomp}[2]{\left\{#1\left|#2\right.\right\}}

\newcommand{\Forall}[2]{\forall #1.\,#2}
\newcommand{\Exists}[2]{\exists #1.\,#2}
\newcommand{\MilnerF}{\mathcal F}

\newcommand{\eqnstop}{\textrm{.}}

\newcommand{\compose}{\circ}

\newcommand{\gproofnotestext}{Notes}

\theoremstyle{definition}

\newcommand{\gindexed}{\mathprefix{indexed}}
\newcommand{\galgtext}{\mathprefix{algebra}}
\newcommand{\gcochain}{\mathprefix{cochain}}
\newcommand{\gchain}{\mathprefix{chain}}
\newcommand{\gpresentable}{\mathprefix{presentable}}

\newcommand{\card}{\mathsf{card}}

\newcommand{\CatA}{\mathcal{C}}

\newcommand{\ArbBehA}{B}
\newcommand{\BehA}{B}
\newcommand{\BehAa}{B'}
\newcommand{\LiftedBehA}{\bar\BehA}

\newcommand{\ArbOrdA}{\alpha}
\newcommand{\iso}{\cong}

\newcommand{\SmallS}{{\mathcal S}}
\newcommand{\srelation}{\mathprefix{relation}}
\newcommand{\spowerset}{\mathprefix{powerset}}
\newcommand{\SPow}[1]{P_{#1}}

\newcommand{\ArbCarrierA}{X}
\newcommand{\ArbCarrierB}{Y}
\newcommand{\ArbCarrierC}{Z}

\newcommand{\CoalgA}{h}
\newcommand{\CoalgB}{k}
\newcommand{\CoalgC}{z}

\newcommand{\RelSeqObjA}[1]{R_{#1}}
\newcommand{\RelSeqHJObjA}[1]{R^{\mathrm{HJ}}_{#1}}
\newcommand{\RelSeqAMObjA}[1]{R^{\mathrm{AM}}_{#1}}
\newcommand{\RelSeqMapA}[2]{r_{#1,#2}}

\newcommand{\ArbOrdB}{\beta}

\newcommand{\ArbLimOrdA}{\lambda}

\newcommand{\SetA}{X}
\newcommand{\SetB}{Y}

\newcommand{\SetZ}{Z}

  \newcommand{\ArbMorphA}{f}

\newcommand{\RelCat}{\mathbf{Rel}_\CatA}

\newcommand{\RPOrder}{\RelCat(\ArbCarrierA,\ArbCarrierB)}
\newcommand{\BehRPOrder}{\RelCat(\BehA\ArbCarrierA,\BehA\ArbCarrierB)}
\newcommand{\pleq}
{\leq}
\newcommand{\pgeq}
{\geq}

\newcommand{\RelA}{R}
\newcommand{\RelB}{S}
\newcommand{\RelAa}{R'}
\newcommand{\relA}{r}

\newcommand{\monoto}{\rightarrowtail}
\newcommand{\coverto}{\rightarrowtriangle}

\newcommand{\RelOP}{\Phi^{\mathrm{HJ}}}
\newcommand{\RelCOP}{\Phi^{\mathrm{AM}}}
\newcommand{\poinA}{i}
\newcommand{\poinB}{j}

\newcommand{\pbackcorner}
{\ar@{}[r]_(.1){\txt{\mbox{}\vspace{-12pt}\\\hspace{9pt}\LARGE{$\lrcorner$}}}}

\def\doi{7 (1:13) 2011}
\lmcsheading%
{\doi}
{1--21}
{}
{}
{Feb.~09, 2010}
{Mar.~29, 2011}
{}

\begin{document}

\title[Relating coalgebraic notions of bisimulation]{Relating coalgebraic notions of bisimulation\rsuper*}

\author[S.~Staton]{Sam Staton}	
\address{Laboratoire PPS, Universit\'e Paris 7}	
\email{sam.staton@cl.cam.ac.uk}  
\thanks{Research supported by EPSRC grants GR/T22049/01 and EP/E042414/1
and ANR Project CHOCO}	


\titlecomment{{\lsuper*}The material in this article is summarized in Sections 1--5 of
  my extended abstract in CALCO'09 \cite{staton-calco09}.}


\begin{abstract}
  The theory of coalgebras, for an endofunctor on a category, has been
  proposed as a general theory of transition systems.  We investigate
  and relate four generalizations of bisimulation to this setting,
  providing conditions under which the four different generalizations
  coincide.
  We study transfinite sequences whose limits are the greatest
  bisimulations.
\end{abstract}

\keywords{Bisimulation, Coinduction, Category theory, Coalgebra}
\subjclass{F.3.2, G.2.m}

\maketitle


\section*{Introduction}

\noindent Notions of bisimulation
play a central role in the theory of transition systems.
The theory of coalgebras provides a setting
in which different notions of system can be
understood at a general level. 
In this article I investigate notions of bisimulation at this general 
level.

To explain the generalization from transition systems to coalgebras,
we begin with the traditional presentation of a 
labelled transition system,
\[
\left(\ArbCarrierA,\ \text{\small (}\!\!\rightarrow_\ArbCarrierA\!\!\text{\small )}\,\subseteq\ArbCarrierA\times\CCSLabels\times\ArbCarrierA\right)\]
(for some set $\CCSLabels$ of labels).
A labelled transition system
can be considered `coalgebraically' as a set $\ArbCarrierA$ 
of states equipped with a 
function 
${\ArbCarrierA\to\SetPow(\CCSLabels\times\ArbCarrierA)}$,
into the powerset of $(\CCSLabels\times\ArbCarrierA)$,
assigning to each state $x\in\ArbCarrierA$ the set
$\{(l,x')~|~x\xrightarrow l_\ArbCarrierA x'\}$.
Generalizing, 
we are led to consider an arbitrary 
category~$\CatA$
and an endofunctor $\ArbBehA$ on it;
then a coalgebra is an object~%
$\ArbCarrierA\in\CatA$ of `states',
and a `next-state' morphism~%
${\ArbCarrierA\to\BehA(\ArbCarrierA)}$.

\subsection*{Coalgebras in different categories.}
Coalgebras appear as generalized transition systems in various settings.
For instance: 
transition systems
for name and value passing process calculi
have been studied in terms of 
coalgebras in
categories of presheaves~(e.g.~\cite{ft-namepassing,gyv-open,staton-thesis});
probabilistic transition systems have been modelled by 
coalgebras for a probability-distribution monad~(e.g.~\cite{bsv-probtypes,vr-bisimprob});
descriptive frames and concepts from modal logic
have been studied in terms of coalgebras over Stone 
spaces~(e.g.~\cite{a-cooksnonwellfounded,bfv-vietoris,kkv-stonecoalg});
basic process calculi with recursion have been modelled
using coalgebras over categories of domains~\cite{k-recbialg,p-bialgrec};
and stochastic transition systems
have been studied in terms of 
coalgebras 
over metric and measurable 
spaces~(see e.g.~\cite{bhmw-pseudometric,ddlp-bisimcocong,v-finalmeas,vr-bisimprob}). 
Finally, there are questions about the conventional theory of 
labelled transition systems in a more constructive universe of 
sets~(e.g.~\cite{bm-nonwellfound}).

\subsection*{Notions of bisimulation.}
Once coalgebras are understood as generalized transition systems,
we can consider bisimulation relations for these systems.
Recall that, for labelled transition systems
$(\SetA,\rightarrow_\SetA)$ and $(\SetB,\rightarrow_\SetB)$, a relation 
$\RelA\subseteq\SetA\times\SetB$ is a \emph{bisimulation}
if, whenever $x\,\RelA\,y$, then for all $l\in L$:
\begin{enumerate}[$\bullet$]
\item For $x'\in X$, if $x\xrightarrow l_X x'$ then there is $y'\in Y$ such that
$y\xrightarrow l_Y y'$ and $x'\,\RelA\,y'$;
\item For $y' \in Y$, if $y\xrightarrow l_Y y'$ then there is $x'\in X$ such that
$x\xrightarrow l_X x'$ and $x'\,\RelA\,y'$.
\end{enumerate}
How should the notion of bisimulation be generalized to the case
of coalgebras for endofunctors on arbitrary categories?
In this article, 
we identify four notions of bisimulation that 
have been proposed
in the coalgebraic context.
\begin{enumerate}[(1)]
\item A relation
  over which a suitably compatible
  transition structure can be defined, as proposed by 
  Aczel and Mendler~\cite{am-finalcoalg};
\item A relation that is compatible for 
  a suitable `relation-lifting' of the endofunctor, as proposed
  by Hermida and Jacobs~\cite{hj-indcoindfib};
\item A relation satisfying a `congruence' condition, proposed by
  Aczel and Mendler~\cite{am-finalcoalg} and used to obtain their 
  general final
  coalgebra theorem;
\item A relation which is the kernel of 
  a common compatible refinement of the two systems.
\end{enumerate}
The four notions coincide for the particular case of 
labelled transition systems.
Under certain 
conditions, the notions are related in the more general 
setting of coalgebras. 

\subsection*{Relationship with the terminal sequence.}
Various authors have constructed terminal coalgebras as a limit of a transfinite 
sequence; the initial part of the sequence is:
\[
1\xleftarrow{!} B(1)\xleftarrow{B(!)} B(B(1))\xleftarrow{B(B(!))}B(B(B(1)))\leftarrow\ \dotsi\ \leftarrow\ \dotsi\ \quad\text.
\]
Of the notions of bisimulation mentioned above, notions~(2) and~(3) 
can often be characterized as post-fixed points of a monotone operator~$\Phi$ on a 
lattice of relations. In this setting, by Tarski's fixed point theorem, 
there is a maximum 
bisimulation (`bisimilarity'). It is given explicitly
as a limit of a transfinite sequence; the initial part of the sequence is:
\[
X\times Y\ \supseteq\ \Phi(X \times Y)\ \supseteq\ \Phi(\Phi(X\times Y))\ \supseteq\ \Phi(\Phi(\Phi(X\times Y)))\ \supseteq\dots\supseteq\dotsi
\]
starting with the maximal relation, that relates everything.  
Under certain conditions, the steps of the terminal coalgebra sequence 
are precisely related with the steps of this relation refinement sequence.

\subsection*{Other approaches not considered.}
In this article we are concerned with internal relations 
between the state objects of 
two fixed coalgebras. A relation is itself an object of the base category.

Some authors (e.g.~\cite{ddlp-bisimcocong}) are concerned 
with defining an equivalence relation
on the class of all coalgebras, by setting two coalgebras as bisimilar 
if there
is a span of surjective homomorphisms between them.
Others work with relations as bimodules (e.g.~\cite{bhmw-pseudometric,r-relators,w-omegacat}).
We will not discuss these approaches here.

\subsection*{Acknowledgements.}
It has been helpful to discuss the material in this article with numerous people over the last eight years, particularly Marcelo Fiore. Benno van den Berg gave some advice on algebraic set theory. 
Many of the results in this article are well-known in the case where
$\CatA=\Set$. In other cases, some results may be folklore; I have
tried to ascribe credit where it is due.

\section{Coalgebras: Definitions and examples}
\label{sec:coalgexamples}

\noindent Recall the definition of a coalgebra for an endofunctor:
\begin{defi}
Consider an endofunctor~$\ArbBehA$ 
on a category~$\CatA$.
A \emph{\gcoalgtext\ArbBehA}\ is given by an 
object $\ArbCarrierA$ of $\CatA$
together with morphism
${\ArbCarrierA\to\ArbBehA(\ArbCarrierA)}$
in $\CatA$.

A \emph{homomorphism} of \gcoalgtext\ArbBehA s, 
from $(\ArbCarrierA,\CoalgA)$ to 
$(\ArbCarrierB,\CoalgB)$, is a 
morphism ${\ArbMorphA:\ArbCarrierA\to\ArbCarrierB}$ 
that respects the coalgebra structures, i.e.\ 
such that 
${\ArbBehA\ArbMorphA\compose\CoalgA
=\CoalgB\compose\ArbMorphA}$.
\end{defi}

\subsection{Examples}
\label{sec:examples}
We collect some examples of concepts that 
arise as coalgebras for endofunctors.
For further motivation, 
see \cite{a-introcoalg,jr-coalgtutorial,r-univcoalg}.

\subsubsection*{Coinductive datatypes.}
Coinductive datatypes can be understood in terms of
coalgebras for polynomial endofunctors.
A polynomial endofunctor on a category with sums and products
is a functor of the following form.
(See e.g.~Rutten~\cite[\textcitesec~10]{r-univcoalg}.)
\[
\SetA\ \mapsto\  \sum_{i\in I} A_i\times \SetA^{n_i}
\]
(Here, each $A_i$ is an object of the category, 
  and each $n_i$ is a natural number.)

\subsubsection*{Transition systems.}
In the introduction we discussed 
the correspondence between labelled transition systems 
and coalgebras for the endofunctor
$\SetPow(\CCSLabels\times(-))$.
Here, $\SetPow$ is the powerset functor, that acts by direct image.
For finite non-determinism, and image-finite transition 
systems, one can instead consider the endofunctor
\[\SetPowF(\CCSLabels\times(-))\]
where $\SetPowF$ is the finite powerset functor, the free semilattice.

\subsubsection*{Transition systems in toposes and name-passing calculi.}
Recall that a topos is a category with finite limits and a powerobject
construction. By definition, the powerobjects classify relations, and
so the coalgebraic characterization
of labelled transition systems
is relevant in any topos.

In process calculi such as the \picalculus~\cite{mpw-pi}, transitions
occur between terms with free variables, and those free variables play
an important role.  Conventional labelled transition systems in the
category of sets are inadequate for such calculi. Instead, one can
work in a category of covariant presheaves $(\Ccat\to\Set)$.  Various
categories have been proposed for~$\Ccat$. We will focus on two 
examples: the
category $\Icat$ of finite sets and injections between them, and the category
$\Fnecat$ of non-empty finite sets and all functions between them.
More sophisticated models of process calculi are found by taking presheaves over more elaborate categories
(see e.g.~\cite{bm-coalgreact,bk-pilogic,gyv-open,s-lics08}).

In this setting, the object $X$ of states is no longer a set,
but a presheaf. For instance, if $X\colon \Icat\to\Set$, 
we think of $X(\objC)$ as the set of states involving the free variables
in the set~$\objC$, and the functorial action of $X$ describes 
injective renaming of states.

The appropriate endofunctor on these presheaf categories 
typically has the following form
\begin{equation}
\label{eqn:npbeh}
\SetPow(B'(-))
\qquad\text{where e.g.}~B'(-)=(N \times (-)^N\, + \,N \times N \times (-)\, +\, (-))
\end{equation}
with the summands of $B'$ 
representing input, output, and silent actions
respectively. The presheaf $N$ is a special object of names.
The interesting question is: what is $\SetPow$?
\begin{enumerate}[$\bullet$]
\item
A natural choice 
is to let $\SetPow$ be the powerobject functor in the presheaf
topos $\PshfC$.
For any presheaf ${Y\in\PshfC}$, and any object~$\objC\in\Ccat$,
$(\SetPow(Y))(\objC)$ is the set of sub-presheaves of $(\Ccat(\objC,-)\times Y)$.
A coalgebra $X\to\SetPow(\BehA'(X))$ 
is a natural transformation between presheaves,
that assigns a behaviour to each state $x\in X(\objC)$, for $\objC\in \Ccat$.
This behaviour is not only a set of future states for $x$, 
but also the future states of $Xf(x)$ for any morphism
$f\colon \objC\to \objC'$ in $\Ccat$.

To understand this more formally, 
recall that the powerobject $\SetPow$ classifies relations in the following sense. 
For presheaves $X$ and $Y$ in $\PshfC$, 
there is a bijective 
correspondence between natural transformations $r\colon X\to\SetPow(Y)$ 
and subpresheaves $R\subseteq X\times Y$.
A subpresheaf $R\subseteq X\times Y$ 
determines a natural transformation $r\colon X\to\SetPow(Y)$;
for $\objC\in \Ccat$ and $x\in X(\objC)$, 
we have 
$r_\objC(x)\in(\SetPow(Y))(\objC)$:
\[
(r_\objC(x))(\objC')=\{(f,y)~|~(Xf(x),y)\in R(\objC')\}
\quad\text.
\]
\item
The powerobject $\SetPow(X)$ accommodates infinite branching transition systems.
To focus on finite branching,
we can find a `finite' subfunctor of $\SetPow$.

The approach taken by Fiore and Turi~\cite{ft-namepassing}
is to let $\SetPow$ be the free semi-lattice (henceforth~$\SetPowF$).
This is sometimes called `Kuratowski finiteness'.
For any presheaf ${Y\in\PshfC}$, and any object $\objC\in\Ccat$,
$(\SetPowF(Y))(\objC)$ is the set of finite subsets of
$Y(\objC)$.
We have an natural monomorphism $i_Y\colon \SetPowF(Y)\monoto \SetPow(Y)$ 
into the full powerobject:
for $\objC\in\Ccat$, $S\subseteq (Y(\objC))$,
we define a subpresheaf $i_{Y,\objC}(S)\in (\SetPow(Y))(\objC)$:
\[
(i_{Y,\objC}(S))(\objC')=\{(f,Yf(y))~|~f\colon \objC\to \objC',\,y\in S\}\quad\text.
\]
\item 
The free semilattice is too naive on the presheaf category
$\PshfFne$.
For example, the $\pi$-calculus process
$(\bar a~|~b)$ cannot perform a $\tau$-step,
but it can perform a $\tau$-step after the substitution
$\{a\mapsto b,~b\mapsto b\}$.
The construction $\SetPowF(X)$ is too
small to allow this information to be recorded.
Indeed, the $\pi$-calculus can be described 
as a coalgebra for the functor 
$\SetPow(\BehA'(-))$ on $\PshfFne$,
but this coalgebra does not factor through the free semilattice,
$\SetPowF(\BehA'(-))$.

In this situation, a more appropriate finite powerset
is the sub-join-semilattice of 
the powerobject $\SetPow(Y)$ that is generated by
the partial map classifier. 
We will write $\SetPowPF(Y)$ for this ---
Freyd~\cite{f-numerology} writes~$\tilde K$.
I gave an algebraic description of 
this construction in \cite{staton-calco09}, and
it has been used by Miculan in his model of the fusion calculus~\cite{m-fusion}.
\end{enumerate}

\subsubsection*{Frames in modal logic.}
Let the base category $\CatA$
be the category of Stone spaces and continuous maps.  (A Stone space
is a compact Hausdorff space in which the clopen sets form a basis.)
Let $K(X)$ be the space of compact subsets of $X$, with the finite
(aka Vietoris) topology. The construction~$K$ is made into a functor,
acting by direct
image. 
Coalgebras for $K$ can be understood as descriptive general frames.
Just as 
the category of Stone spaces is dual to the category of Boolean algebras,
the category of $K$-coalgebras is dual to the category 
of modal algebras (i.e., Boolean algebras
equipped with a meet-preserving operation)
--- see e.g.~\cite{a-cooksnonwellfounded,kkv-stonecoalg}.

\subsubsection*{Powersets in algebraic set theory.}
A general treatment of powersets is suggested by the algebraic set 
theory of Joyal and Moerdijk~\cite{jm-ast}.
A model of algebraic set theory is a category~$\CatA$ together with a
class of `small' maps~$\SmallS$ in~$\CatA$, all subject to certain
conditions. 
An intuition is that a map $f\colon X\to Y$ is small if its fibres
$f\inv(y)$ are all small.

In such a situation, an \emph{\srelation\SmallS}\ is a 
relation ${\RelA\subseteq \SetA\times\SetB}$ 
for which the projection $\RelA\to\SetA$ is in $\SmallS$.
An endofunctor $\SPow\SmallS$ on $\CatA$ is said to be 
the \emph{\spowerset\SmallS} 
if there is an \srelation\SmallS\ 
$(\ni_\SetB)\subseteq\SPow\SmallS(\SetB)\times\SetB$
inducing a
bijective correspondence between 
\srelation\SmallS s $(\RelA\subseteq \SetA\times \SetB)$ and 
morphisms $\SetA\to \SPow\SmallS(\SetB)$.
Further details are given in the appendix.

These ideas cater for the notions of power set discussed so far.  For instance:
\begin{enumerate}[$\bullet$]
\item 
Let $\CatA$ be the category of sets, in the classical sense,
and say that a function ${f:X\to Y}$ is small if for every $y\in Y$ the set
$f\inv(y)$, i.e. $\{x\in X~|~f(x)=y\}$, is finite. 
This class $\SmallS$ of maps satisfies all the axioms for small maps
given in the appendix.
An \srelation\SmallS\ is precisely 
an image-finite one, and the \spowerset\SmallS\ is the finite powerset.
\item \newcommand{\IndCatA}{\mathbb{C}}%
For a presheaf category $[\IndCatA,\Set]$, 
the free semilattice construction~$\SetPowF$ 
is an \spowerset\SmallS\ where
$\SmallS$ is the class of natural transformations between
presheaves, 
$\phi\colon X\to Y$,
such that (i)~for each $\objC\in\IndCatA$, $y\in Y(\objC)$,
the set $\{x\in X(\objC)~|~f_\objC(x)=y\}$ is finite, and
(ii)~each naturality square is a weak pullback,
i.e., if $\phi_{\objC'}(x')=Yf(y)$ then there is $x\in X(\objC)$ 
such that $\phi_\objC(x)=y$ and $Xf(x)=x'$:
\begin{equation}
\xymatrix{X(\objC)\ar[r]^{\phi_\objC}\ar[d]_{X(f)} &Y(\objC)\ar[d]^{Y(f)}
\\
X(\objC')\ar[r]_{\phi_{\objC'}}& Y(\objC')\text.}
\label{eqn:openmap}
\end{equation}
This class~$\SmallS$ of morphisms always satisfies 
Axioms A1--A6 and~A9 
for small maps, but not (M): monos are not 
small unless $\IndCatA$ is a groupoid.

In the presheaf category
$\PshfFne$, 
the free semilattice generated by the partial map classifier,
$\SetPowPF$, is more liberal: it
classifies the natural transformations between
presheaves
that satisfy the finiteness condition~(i) but the requirement on
naturality~(ii)
is weakened to the situation when $f$ is an injection.
This class of morphisms satisfies 
all the axioms for small maps, including~(M).
(This argument is quite specific to $\Fnecat$.)
\item
Let $\CatA$ be the category of Stone spaces, and let $K(X)$ be the 
space of compact subsets of $X$.
Recall that a continuous map is \emph{open}
if the direct image of an open set is open.
The class $\SmallS$ of open maps in the category of Stone spaces 
satisfies Axioms A1--A6 and~A9 for small maps.
The evident relation ${\ni_X}\subseteq K(X)\times X$ is 
an \srelation\SmallS, and 
exhibits each $K(X)$ as an \spowerset\SmallS.

(Although the Vietoris construction can be considered over more general spaces,
the characterization of $K$ as an \spowerset\SmallS\ is 
specific to Stone spaces. 
This raises the question of how best to treat
more expressive positive/topological set 
theories~\cite{le-topsettheory} 
in an algebraic setting.)
\end{enumerate}
\subsubsection*{Probabilistic transition systems.}
For any set~$\ArbCarrierA$, let~$\DistF(\ArbCarrierA)$ be the set 
of sub-probability distribution functions on~$\ArbCarrierA$,
\viz, 
functions from~$d\colon \ArbCarrierA\to[0,1]$ into the unit interval
for which $\{x\in\ArbCarrierA~|~d(x)\neq 0\}$ is finite and
$\sum_{x\in\ArbCarrierA}d(x)\leq 1$.
This construction extends to an endofunctor on~$\Set$, with 
the covariant action given by summation. 
Coalgebras for~$\DistF$ are discrete 
probabilistic transition systems~%
\cite{bsv-probtypes,vr-bisimprob}.

\subsubsection*{Systems where the state space has more structure}
For continuous stochastic systems, researchers have investigated 
coalgebras for probability distribution 
functors on categories of metric or measurable spaces (see e.g.~\cite{bhmw-pseudometric,ddlp-bisimcocong,v-finalmeas,vr-bisimprob}).

For recursively defined systems, 
it is reasonable to investigate coalgebras for 
powerdomain constructions on a category of 
domains~(in the bialgebraic context, see~\cite{k-recbialg,p-bialgrec}).

Levy and Worrell~\cite{levy-sim,w-omegacat} 
have considered endofunctors on categories of
preorders, posets, and categories enriched
in quantales, in their investigations of similarity.

\section{Bisimulation: four definitions}
\label{sec:notions}

\noindent We now recall four notions of bisimulation
on the state spaces of coalgebras.
The four notions generalize the standard notion of bisimulation
for labelled transition systems
(\ie\ coalgebras for $\SetPow(L \times (-))$ on $\Set$),
due to Milner~\cite{m-ccs}
and Park~\cite{p-bisim}.
For all four notions, the maximal bisimulation is the usual notion of 
strong bisimilarity for labelled transition systems.

To some extent, the different notions of bisimulation
have arisen from the examples in \secref{sec:examples},
as authors sought coalgebraic notions of bisimulation that were 
appropriate to the base category and endofunctor under consideration,
as well as to the intended applications.
For name-passing calculi, coalgebraic bisimulation 
can be used to capture the bisimulations of 
Milner, Parrow and Walker~\cite{mpw-pi} 
and also the open bisimulation of Sangiorgi~\cite{s-openbisim} 
(see e.g.~\cite[Sec.~6]{ft-namepassing,staton-calco09});
for discrete-space probabilistic systems, coalgebraic bisimulation can 
describe
the probabilistic bisimulation of Larsen and Skou~\cite{ls-probbisim}
(see e.g.~\cite{vr-bisimprob}).

\subsubsection*{Relations in categories.}
For objects~$\ArbCarrierA$, $\ArbCarrierB$ of a category~$\CatA$, we
let $\RPOrder$ be the preorder of relations, \viz\ jointly-monic spans
${\ArbCarrierA\leftarrow\RelA\rightarrow\ArbCarrierB}$, where
$\RelA\pleq\RelAa$ if $\RelA$ factors through $\RelAa$.  
When $\CatA$ has products, the preorder $\RPOrder$ coincides with the 
preorder of monos into $(\ArbCarrierA\times\ArbCarrierB)$.
Relations are most well-behaved in regular categories (see Appendix).

\subsubsection*{Context.}
In this section
we fix a category $\CatA$ 
and consider an endofunctor~$\ArbBehA$ 
on~$\CatA$.
We fix two \gcoalgtext\ArbBehA s,
${\CoalgA:\ArbCarrierA\to\ArbBehA\ArbCarrierA}$
and
${\CoalgB:\ArbCarrierB\to\ArbBehA\ArbCarrierB}$.

\subsection[Hello]{The lifting-span bisimulation of Aczel and Mendler
\cite{am-finalcoalg}}
This notion is categorically the simplest. It directly dualizes the concept
of congruence from universal algebra.
\begin{defi}
A relation
${\RelA\in\RPOrder}$
is an
\emph{\AMbisimtext}\ between 
$(\ArbCarrierA,\CoalgA)$ and
$(\ArbCarrierB,\CoalgB)$
if there exists a \gcoalgtext\ArbBehA\ structure
on $\RelA$ that lifts it to 
a span of coalgebra homomorphisms, as in the following diagram.
\[
\xymatrix{
\ArbCarrierA\ar[d]_\CoalgA
&
\RelA\ar[l]\ar[r]
\ar@{..>}[d]^\exists
&
\ArbCarrierB\ar[d]^\CoalgB
\\
\ArbBehA\ArbCarrierA
&
\ArbBehA\RelA
\ar[r]\ar[l]
&
\ArbBehA\ArbCarrierB
}
\]
\end{defi}

\hide
{An equivalent definition is possible if $\CatA$ has products:
a mono
${\RelA\monoto\ArbCarrierA\times\ArbCarrierB}$ 
is
a \AMbisimtext\ if there
is a coalgebra structure
${\relA:\RelA\to\ArbBehA\RelA}$ making the following 
diagram commute.
\begin{equation}
\label{dgm:bisimprod}
\xymatrix{
  \RelA
  \ar[d]_\relA
  \ar@{ >->}[rr]
  &&\ArbCarrierA\times\ArbCarrierB
  \ar[d]^{\CoalgA\times\CoalgB}
  \\
  \BehA\RelA
  \ar[r]
  &
  \BehA(\ArbCarrierA\times\ArbCarrierB)
  \ar[r]
  &
  \BehA\ArbCarrierA\times\BehA\ArbCarrierB
}
\end{equation}
}
\subsection{The relation-lifting bisimulation of Hermida and Jacobs \cite{hj-indcoindfib}}
The following is a simplification of the bisimulation of Hermida and
Jacobs, who work in a more general fibrational setting.
\label{sec:fixedpoints}
\label{sec:finalbisimexist}
\newcommand{\LimOP}{L}
\newcommand{\limOP}{l}
\begin{defi}
Let~$\CatA$ have products and images. (See the Appendix for definition.)
For any relation 
$\RelA\in\RPOrder$,
we define the relation
$\LiftedBehA\RelA\in\BehRPOrder$ to be the 
image of the composite morphism
${\BehA\RelA
\to\BehA(\ArbCarrierA\times\ArbCarrierB)
\rightarrow
\BehA\ArbCarrierA\times\BehA\ArbCarrierB}
$.
(The construction $\LiftedBehA$ is called 
the ``relation lifting'' of $\BehA$.)

A relation $\RelA$ in $\RPOrder$ is an \emph{\HJbisimtext}\ if 
there is a morphism $\RelA\to\LiftedBehA(\RelA)$ making the following
diagram commute.
\[
\xymatrix{
\ArbCarrierA\ar[d]_\CoalgA&\RelA\ar[l]\ar[r]\ar@{..>}[d]&\ArbCarrierB\ar[d]^\CoalgB
\\
\BehA\ArbCarrierA&\LiftedBehA\RelA\ar[l]\ar[r]&\BehA\ArbCarrierB
}
\]
\end{defi}
\noindent 
When $\CatA$ has pullbacks, let 
$\RelOP(\RelA)$ be the following pullback:
\[
\xymatrix@R-6pt{
  \RelOP(\RelA)\ar[r]\ar@{ >->}[d]\pb{-45}&\LiftedBehA\RelA\ar@{ >->}[d]
  \\
  \SetA\times\SetB\ar[r]_-{\CoalgA\times\CoalgB}&\BehA\SetA\times\BehA\SetB
}
\]
By definition, a relation $\RelA$ is an
\emph{\HJbisimtext} if and only if $\RelA\leq \RelOP(\RelA)$.
\begin{prop}
\label{prop:relopmonotone}
  The operator $\RelOP$ on $\RPOrder$ is monotone.\qed
\end{prop}
\newcommand{\EltA}{x}%
\newcommand{\EltAa}{{x'}}%
\newcommand{\EltB}{y}%
\newcommand{\EltBb}{{y'}}%
For illustration, we briefly return to  the situation
of transition systems,
where 
${\CatA=\Set}$ 
and ${\ArbBehA=\SetPow(\CCSLabels\times-)}$.
For any relation
${\RelA\in\RPOrder}$,
the refined relation 
${\RelOP(\RelA)}$ in ${\RPOrder}$
is the set of all pairs 
$(\EltA,\EltB)\in\ArbCarrierA\times\ArbCarrierB$
for which 
\begin{align*}
\text{(i)\ }&
\Forall{(\CCSLab,\EltAa)\in\CoalgA(\EltA)}
{\Exists{\EltBb\in\ArbCarrierB}
  {(\CCSLab,\EltBb)\in\CoalgB(\EltB)
    \text{ and }
    (\EltAa,\EltBb)\in\RelA}}
\text{ ;}
\\
\text{(ii)\ }&
\Forall{(\CCSLab,\EltBb)\in\CoalgB(\EltB)}
{\Exists{\EltAa\in\ArbCarrierA}
  {(\CCSLab,\EltAa)\in\CoalgA(\EltA)
    \text{ and }
    (\EltAa,\EltBb)\in\RelA}}
\ \eqnstop
\end{align*}
Thus 
the operator~%
$\RelOP$ is the construction~%
$\MilnerF$ considered 
by Milner~\cite[\textcitesec~4]{m-csa}.

\subsection{The congruences of Aczel and Mendler
\cite{am-finalcoalg}}
\begin{defi}
A relation $\RelA$ in $\RPOrder$
is an \emph{\AMprecong} if 
for every cospan
$(\SetA\xrightarrow\poinA\SetZ\xleftarrow\poinB \SetB)$,
\[
\xymatrix@R-18pt@C-6pt{
&&\SetA\ar[dr]^\poinA&
&&
&\SetA\ar[r]^-\CoalgA&\BehA\SetA\ar[dr]^{\BehA(\poinA)}
\\
\text{if}&
\RelA\ar[ur]\ar[dr]&&\SetZ
&\text{commutes then so does}&
\RelA\ar[ur]\ar[dr]&&&\BehA\SetZ
\\
&
&\SetB\ar[ur]_\poinB&
&&
&\SetB\ar[r]_-\CoalgB&\BehA\SetB\ar[ur]_{\BehA(\poinB)}
}
\]
\end{defi}
\noindent 
The definition might appear clumsy and unmotivated,
but \AMprecong s are 
of primary interest because of their connection with 
terminal coalgebras in a general setting,
as will become clear in Theorem~\ref{thm:relate}.

If $\CatA$ has pushouts, then it is sufficient to check 
the case where $\SetZ$ is the pushout of $\RelA$.
If $\CatA$ also has pullbacks, let $\RelCOP(\RelA)$ be the pullback 
of the cospan 
\[
\SetA\xrightarrow\CoalgA\BehA\SetA\xrightarrow{\BehA\poinA}\BehA\SetZ
\xleftarrow{\BehA\poinB}\BehA\SetB\xleftarrow\CoalgB\SetB\quad\text.
\]
By definition, a relation $\RelA$ is an
\emph{\AMprecong} if and only if $\RelA\leq \RelCOP(\RelA)$.

\begin{prop}
\label{prop:COPmonotone}
  The operator $\RelCOP$ on $\RPOrder$ is monotone.\qed
\end{prop}
\noindent
(Note that $\RelCOP$ is different from $\RelOP$, even when $\BehA$ is
the identity functor on $\Set$.)

Our definition differs from that of \cite{am-finalcoalg}
in that we consider relations between different coalgebras.
The connection is as follows: if $(\SetA,\CoalgA)=(\SetB,\CoalgB)$,
then an equivalence relation is an \AMprecong\ 
exactly when it is a congruence in the sense of Aczel and 
Mendler~\cite{am-finalcoalg}
(see \secref{sec:equiv}).

\subsection{Terminal coalgebras and kernel-bisimulations}
Many authors have argued that 
when the category of coalgebras has a terminal object, equality in 
the terminal coalgebra is the right notion of bisimilarity.
Suppose that the category $\CatA$ 
has pullbacks, 
and suppose for a moment that there is a terminal \gcoalgtext\BehA,
$(Z,z)$.
This induces a relation in $\RPOrder$
as the pullback of the unique terminal morphisms
$X\to Z\leftarrow Y$. The relation is sometimes called `behavioural equivalence'.

We can formulate a related notion of bisimulation without 
assuming that there is a terminal 
coalgebra.
\begin{defi}
Let $\CatA$ have pullbacks. 
A relation $\RelA$ is a \emph{kernel-bisimulation} 
if there is a 
\gcoalgtext\BehA\ $(\ArbCarrierC,\CoalgC)$
and a cospan of homomorphisms,
${(\ArbCarrierA,\CoalgA)\rightarrow
(\ArbCarrierC,\CoalgC)\leftarrow
(\ArbCarrierB,\CoalgB)}$,
and $\RelA$ is the pullback 
of ${(\ArbCarrierA\rightarrow\ArbCarrierC\leftarrow\ArbCarrierB)}$.
\end{defi}
Aside from the great many works involving terminal coalgebras,
various authors (e.g.~\cite{gs-typecoalg,k-recbialg,kurz-thesis})
have used kernel-bisimulations (though not by this name).
(The term `cocongruence' is sometimes used
to refer directly to the cospan involved.)

\section{Properties of endofunctors}
\label{sec:propendo}

\noindent All the definitions of the previous section are relevant when $\CatA$
is a regular category. From a categorical perspective, one might
restrict attention to endofunctors that are regular, i.e.\ that
preserve limits and covers. But none of the 
endofunctors in 
\secref{sec:examples} are regular.  We now recall five weaker conditions that
might be assumed of our endofunctor~%
$\BehA$.

\begin{enumerate}[(1)]
\item 
The image of a relation under a functor need not again be a relation,
and one can restrict attention 
to endofunctors that \emph{preserve relations},
i.e., for which a jointly-monic span is mapped to a jointly-monic span.
\end{enumerate}
The remaining restrictions have to do with weak forms of 
pullback-preservation.
To introduce them, we consider
a cospan~${(A_1\rightarrow Z\leftarrow A_2)}$
in~$\CatA$, and in particular the mediating morphism
$
m\colon\ArbBehA(A_1\times_Z A_2) \to (\ArbBehA A_1)\times_{\ArbBehA Z} (\ArbBehA A_2)
$
from the image of the pullback to the pullback of the image:
\[
\xymatrix@R-10pt@C-5pt{
&A_1\ar[dr]
\\
A_1\times_Z A_2
\pbbb{0}
\ar[ur]^{\pi_1}\ar[dr]_{\pi_2}
&&\SetZ
\\
&A_2\ar[ur]
}\qquad
\xymatrix@R-10pt@C-5pt{
&&\BehA A_1\ar[dr]
\\
\BehA(A_1\times_ZA_2)\ar@{..>}[r]^m
\ar@/^/[rru]^{\BehA\pi_1}\ar@/_/[rrd]_{\BehA\pi_2}
&
{\BehA A_1\times_{\BehA Z}\BehA A_2\hspace{-1cm}}
\ar[ur]\ar[dr]
&&\BehA Z
\\
&&\BehA A_2\ar[ur]}
\]
Here are some conditions on $\BehA$, listed in order of decreasing strength.
\begin{enumerate}[(1)]
\item[(2)]
  $\ArbBehA$ \emph{preserves pullbacks}
  if $m$ is always an isomorphism.
\item[(3)]
  $\ArbBehA$ \emph{preserves weak pullbacks},
  if $m$ is always
  split epi. Gumm~\cite{g-coalgfunctorwpp} describes 
  several equivalent definitions
  of this term.
\item[(4)]
  $\ArbBehA$ \emph{covers pullbacks} 
  if $m$ is always a cover.
  The terminology is due to   \cite{t-relautomata}.
  Note that a split epi is always a cover. 
\item[(5)]
  $\ArbBehA$ \emph{preserves pullbacks along monos}
  if $m$ is an isomorphism when $A_1\to Z$ is monic.
\end{enumerate}
Tying up with relation preservation (1):~%
$\BehA$
preserves pullbacks if and only if 
it preserves relations and covers pullbacks
(see Carboni et al.~\cite[Sec.~4.3]{ckw-twocatgeom}).

\subsection{Relevance of the properties}
\begin{prop}\hfill
\begin{enumerate}[\em(1)]
\item 
  Let $\Psi$ be one of the following properties: 
  relation preservation, pullback preservation, 
  weak pullback preservation, preservation of pullbacks along monos.
  The composition of two endofunctors satisfying~$\Psi$ also satisfies~$\Psi$.
\item If $\BehA$ and $\BehAa$ both cover pullbacks and $\BehAa$ 
  preserves covers, then the composite $(\BehAa\BehA)$ also covers pullbacks.
\qed
\end{enumerate}
\end{prop}
\begin{prop}
Every polynomial endofunctor on an extensive category preserves pullbacks 
and covers.\qed
\end{prop}
Regarding powerset functors from algebraic set theory (see Appendix),
we have the following general results.
\begin{prop}
\label{prop:astwpp}
  Let $\CatA$ be regular, and let~$\SmallS$ be a class of open maps
  in~$\CatA$.
  Suppose that $\SPow\SmallS$ is an \spowerset\SmallS.
\begin{enumerate}[\em(1)]
\item\label{astwpp:one} 
  The functor $\SPow\SmallS$ preserves pullbacks along monomorphisms.
\item \label{astwpp:two} 
  If $\SmallS$ 
  contains all monomorphisms~(M), then~$\SPow\SmallS$ 
  preserves weak pullbacks.
\item\label{astwpp:three} Let $\CatA$ also be extensive, and let $\SmallS$ 
  satisfy the axioms for extensive categories
  (axioms~(\mbox{A1--6})) and also collection~(A9),
  but not necessarily~(M).
  The functor $\SPow\SmallS$ covers pullbacks.
\end{enumerate}
\end{prop}
\begin{proofof}{\propref{prop:astwpp}}
  I will sketch proofs in the set-theoretic notation.
  Translation into categorical language is 
  straightforward.

  For \itemref{astwpp:one}, consider a pullback square, and its image under $\SPow\SmallS$.
  \[
  \xymatrix{
    f^{-1}(Z')\pb{-45}\ar[d]\ar@{ >->}[r]&A\ar[d]^f
    \\
    Z'\ar@{ >->}[r]&Z
  }
  \qquad\qquad
  \xymatrix{
    \SPow\SmallS(f^{-1}(Z'))\ar[d]\ar[r]&\SPow\SmallS(A)\ar[d]^{\SPow\SmallS f}
    \\
    \SPow\SmallS(Z')\ar[r]&\SPow\SmallS(Z)
  }
  \]
  To see that the right-hand square is a pullback, consider
  $S$ in $\SPow\SmallS(A)$, for which the direct image
  $\SPow\SmallS f(S)$ is in $\SPow\SmallS(Z')$;
  then, by definition, $S$ is in $\SPow\SmallS (f^{-1}(Z'))$.

  For items~(\ref{astwpp:two}) and (\ref{astwpp:three}), consider a cospan~${(A_1\rightarrow Z\leftarrow A_2)}$.
  For \itemref{astwpp:two}, we must exhibit a section
  ${s:\SPow\SmallS(A_1)\times_{\SPow\SmallS(Z)}\SPow\SmallS(A_2)\to \SPow\SmallS(A_1\times_Z A_2)}$ of 
  the canonical morphism. For $(S_1,S_2)$ in $\SPow\SmallS(A_1)\times_{\SPow\SmallS(Z)}\SPow\SmallS(A_2)$,
  note that 
  \[
  \forall a_1\in S_1.\ \exists a_2 \in S_2.\ f(a_1)=g(a_2)\quad
  \text{and}\quad
  \forall a_2\in S_2.\ \exists a_1 \in S_1.\ f(a_1)=g(a_2)
  \]
  and let
  \[s(S_1,S_2) = \{ (a_1,a_2) \in (S_1\times S_2)~|~ f(a_1) = g(a_2)\}
  \quad\text.\]
  Here we have used the separation axiom, which is valid 
  when all monomorphisms are small.

  For \itemref{astwpp:three}, we show that the canonical morphism
  \[{\SPow\SmallS(A_1\times_Z A_2)\to \SPow\SmallS(A_1)\times_{\SPow\SmallS(Z)}\SPow\SmallS(A_2)}\]
  is a cover.
  Consider $(S_1,S_2)$ in $\SPow\SmallS(A_1)\times_{\SPow\SmallS(Z)}\SPow\SmallS(A_2)$;
  we must show that there is $S$ in $\SPow\SmallS(A_1\times_Z A_2)$
  whose direct image is $(S_1,S_2)$.
  
  By the (strong) collection axiom, we have 
  $T_1$ in $\SPow\SmallS(A_1\times_Z A_2)$ 
  such that 
  ${\pi_1(T_1)=S_1}$ and 
  ${\pi_2(T_1)\subseteq S_2}$. 
  Similarly, we also have 
  $T_2$ in $\SPow\SmallS(A_1\times_Z A_2)$ such that
  ${\pi_2(T_2)=S_2}$ and 
  ${\pi_1(T_2)\subseteq S_1}$.
  Thus $(T_1\cup T_2)$ is in $\SPow\SmallS(A_1\times_Z A_2)$,
  and its direct image is $(S_1,S_2)$, as required.
  Thus \propref{prop:astwpp} is proved.
\end{proofof}

The free semi-lattice functor $\SetPowF$ on any topos covers
pullbacks, but will only preserve weak pullbacks if the 
topos is Boolean.
In general, the corresponding class of small maps 
does not contain all monomorphisms,
as Johnstone et al.~\cite[Ex.~1.4]{jptww-catcoalg} have observed.
The counterexample of \cite{jptww-catcoalg}
is easily adapted to the 
settings of the presheaf categories for name passing,
correcting oversights in
\cite{ft-namepassing,gyv-open,staton-thesis}.

On the other hand,  
the endofunctor $\SetPowPF$ 
on the presheaf category $\PshfFne$
does preserve weak pullbacks.

The compact-subspace endofunctor $K$ on Stone spaces
covers pullbacks (a consequence of \propref{prop:astwpp})
although it does not preserve weak pullbacks,
as observed by Bezhanishvili et al~\cite{bfv-vietoris}.

The probability distribution functor $\DistF$ on $\Set$ preserves
weak pullbacks~\cite{m-coalglogic,vr-bisimprob}.

More sophisticated continuous settings are problematic. 
Counterexamples to the weak-pullback-preservation of 
probability distributions on measurable spaces 
are discussed in~\cite{v-finalmeas}. 
Plotkin~\cite{p-bialgrec} discusses problems with 
coalgebraic bisimulation in categories of domains:
the convex powerdomain does not even preserve monomorphisms.
The endofunctors on posets that Levy
considers~\cite{levy-sim} 
typically do not preserve monomorphisms either.

\section{Relating the notions of bisimulation}
\label{sec:relatcat}

\noindent The purpose of this section is to relate the four notions of
bisimulation introduced in \secref{sec:notions}.

\begin{thm}Let $\BehA$ be an
  endofunctor on a category $\CatA$ with finite limits and images.
\label{thm:relate}
\begin{enumerate}[\em(1)]
\item\label{relate:one}
Every \AMbisimtext\ is an \HJbisimtext.
\item\label{relate:two}
Every \HJbisimtext\ is an \AMprecong.
\item\label{relate:three}
Every \AMprecong\ is contained in a kernel bisimulation that is an
\AMprecong, provided $\CatA$ has pushouts.
\item\label{relate:four}
Every kernel bisimulation is an \AMbisimtext, provided
$\BehA$ preserves weak pullbacks.
\item\label{relate:five}
Every kernel bisimulation is an \HJbisimtext, provided
$\BehA$ covers pullbacks.
\item\label{relate:six}
Every kernel bisimulation is an \AMprecong, provided
$\BehA$ preserves pullbacks along monos and $\CatA$ is regular.
\item\label{relate:seven}
Every \HJbisimtext\ is an \AMbisimtext, provided either
\begin{enumerate}[(i)]
\item \label{relate:sevena} every epi in $\CatA$ is split, or 
\item \label{relate:sevenb} $\BehA$ preserves relations, or
\item \label{relate:sevenc}$\CatA$ is regular with a class~$\SmallS$ of open maps
  containing all monomorphisms, and 
  there is an \spowerset\SmallS\ $\SPow\SmallS$,
  and $\BehA(-)\iso \SPow\SmallS (\BehA'(-))$, for a relation preserving functor $\BehA'$.
\item \label{relate:sevend} $\CatA$ is a topos and $\BehA\iso P(\BehA'(-))$, where $P$ is the
  powerobject of $\CatA$ and $\BehA'$ is an arbitrary endofunctor.
\end{enumerate}
\end{enumerate}
\end{thm}
In summary:
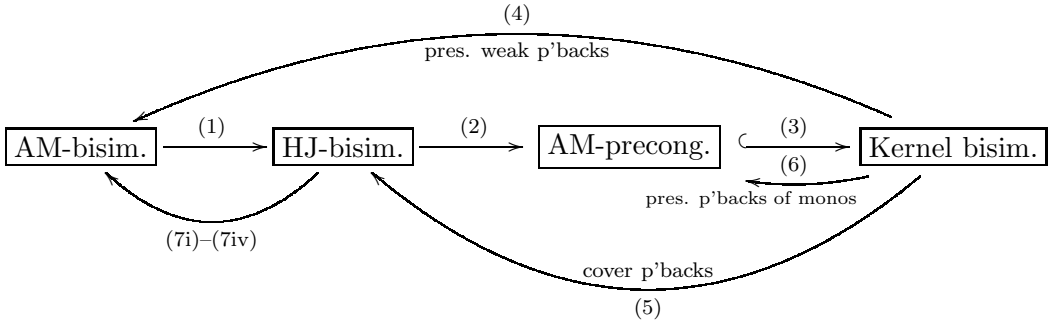
\[\xymatrix@R-40pt@C+15pt{
*+<5pt>{\txt{\framebox{AM-bisim.}}}
\ar[r]^{(1)}
&
*+<5pt>{\txt{\framebox{HJ-bisim.}}}
\ar[r]^(.45){(2)}
\ar@/^1cm/[l]^{\text{(7i)--(7iv)}}
&
*+<12pt>{\txt{\framebox{AM-precong.}}\ }
\ar@{^{(}->}[r]^{(3)}
&
*+<8pt>{\txt{\framebox{Kernel bisim.}}}
\ar@/_1.5cm/[lll]^{\text{pres. weak p'backs}}_{(4)}
\ar@/^1.9cm/[ll]_{\text{cover p'backs}}^{(5)}
\ar@/^.5cm/[l]^(.6){\txt{\tiny pres. p'backs of monos}}_{(6)}
}\]

Note that the different notions of bisimulation are not, in general, the same.
For instance, Aczel and Mendler~\cite[p.~363]{am-finalcoalg} provide 
an example of an endofunctor on $\Set$ for which there is an
\AMprecong\ 
that is not an \AMbisimtext.
Bezhanishvili et al.~\cite[Sec.~4]{bfv-vietoris} demonstrate
that \AMbisimtext\ is different from 
\HJbisimtext\ for the Vietoris construction on Stone spaces.

\subsection{Proof of \thmref{thm:relate}}
Throughout the proof, we fix two \gcoalgtext\ArbBehA s,
${\CoalgA:\ArbCarrierA\to\ArbBehA\ArbCarrierA}$
and
${\CoalgB:\ArbCarrierB\to\ArbBehA\ArbCarrierB}$.

\Itemref{relate:one} is trivial. 

For \itemref{relate:two}, let $\RelA$ be an \HJbisimtext.
Let $(\SetA\to\SetZ\leftarrow\SetB)$ be a cone over 
the span $(\SetA\leftarrow\RelA\to\SetB)$.
Consider the following commuting diagrams.
\[
\text{(a)}\ 
\xymatrix@R-15pt{
&\SetA\ar[r]^(.45)\CoalgA&\BehA\SetA\ar[dr]
\\
\RelA\ar[ur]
\ar[dr]
\ar[r]
&
\LiftedBehA\RelA
\ar[ur]\ar[dr]
&&\BehA\SetZ
\\
&\SetB\ar[r]_(.45)\CoalgB&\BehA\SetB\ar[ur]}
\qquad
\text{(b)}\ 
\xymatrix@R-15pt{
&\BehA\SetA\ar[dr]
\\
\BehA\RelA
\ar[ur]\ar[dr]
&&\BehA\SetZ
\\
&\BehA\SetB\ar[ur]}
\]
The left-hand squares of diagram~(a) say that $\RelA$ is an \HJbisimtext.
The right-hand square of diagram~(a) commutes since 
diagram~(b) commutes, and 
$\BehA\RelA\coverto\LiftedBehA\RelA$
is epi. Thus the whole of~(a) commutes, and $\RelA$ is an \AMprecong.

For \itemref{relate:three}, let $\RelA$ be an \AMprecong. Let 
$(\SetA\to\SetZ\leftarrow\SetB)$ be the pushout 
of the span $(\SetA\leftarrow\RelA\to\SetB)$.
The following diagram commutes; the dotted morphism 
follows from universality of the pushout.
\[
\xymatrix@R-15pt{
&\SetA\ar[r]^(.45)\CoalgA\ar[dr]&\BehA\SetA\ar[dr]
\\
\RelA\ar[ur]
\ar[dr]
&&
\pb{180}
\SetZ\ar@{..>}[r]
&\BehA\SetZ
\\
&\SetB\ar[r]_(.45)\CoalgB\ar[ur]&\BehA\SetB\ar[ur]}
\]
Let 
$(\SetA\leftarrow\RelAa\to\SetB)$
be the pullback of $(\SetA\to\SetZ\leftarrow\SetB)$.
By definition, it is a kernel bisimulation.
Moreover, the pushout of 
$(\SetA\leftarrow\RelAa\to\SetB)$
is~$\SetZ$ again, so~$\RelAa$ is an \AMprecong.

For items~(\ref{relate:four}), (\ref{relate:five}) and (\ref{relate:six}),
let $\RelA$ be a kernel bisimulation, the pullback
of a cospan ${(\ArbCarrierA\to\ArbCarrierC\leftarrow\ArbCarrierB})$, 
for some coalgebra $(\ArbCarrierC,\CoalgC)$.
Note that the following diagram commutes.
\begin{equation}
\label{dgm:kerneltoothers}
\xymatrix@R-15pt{
&\ArbCarrierA\ar[r]^\CoalgA&\BehA\ArbCarrierA\ar[dr]
\\
\RelA\ar[ur]\ar[dr]
&&&\BehA\ArbCarrierC
\\
&\ArbCarrierB\ar[r]^\CoalgB&\BehA\ArbCarrierB\ar[ur]
}
\end{equation}
For \itemref{relate:four}, we must show that $\RelA$ is an \AMbisimtext. 
We construct a coalgebra structure on $\RelA$ by
considering the morphism $\RelA\to\BehA\RelA$ induced 
since $\BehA\RelA$ is a weak pullback, as in the following diagram.
\[
\xymatrix@R-15pt{
&\ArbCarrierA\ar[r]^\CoalgA&\BehA\ArbCarrierA\ar[dr]
\\
\RelA\ar[ur]\ar[dr]\ar@{..>}[r]
&\BehA\RelA\ar[ur]\ar[dr]&&\BehA\ArbCarrierC
\\
&\ArbCarrierB\ar[r]^\CoalgB&\BehA\ArbCarrierB\ar[ur]
}
\]

For \itemref{relate:five}, we must show that $\RelA$ is an \HJbisimtext. This follows from the following fact, which is immediate from the definition
of $\LiftedBehA\RelA$, and which is worth recording:
\begin{fact}
\label{fact:coverpbacks}
  If $\BehA$ covers pullbacks, 
  then
  \[
  \xymatrix@R-20pt@C-10pt{
    &&\SetA\ar[dr]&&&    &\BehA\SetA\ar[dr]
    \\
    \text{if }&
    \RelA\ar[ur]\ar[dr]
    &&
    \SetZ
    &  \text{ is a pullback, so is }
    &\LiftedBehA\RelA\ar[ur]\ar[dr]
    &&
    \BehA\SetZ&\text.
    \\
    &&\SetB\ar[ur]&&&    &\BehA\SetB\ar[ur]
  }
  \] 
%
\end{fact}
\newcommand{\SetV}{V}%
\newcommand{\SetVv}{V'}%

\newcommand{\RelBb}{S'}
\newcommand{\MapA}{f}
\newcommand{\MapAa}{f'}
\newcommand{\MapB}{g}
\newcommand{\MapBb}{g'}
\newcommand{\ProjA}{p}
\newcommand{\ProjAa}{p'}
\newcommand{\ProjB}{q}
\newcommand{\ProjBb}{q'}
\newcommand{\imMapA}{\mathrm{im}(f)}
\newcommand{\imMapB}{\mathrm{im}(g)}
\newcommand{\imProjA}{\mathrm{im}(p)}
\newcommand{\imProjB}{\mathrm{im}(q)}
For \itemref{relate:six}, we must show that $\RelA$ is an \AMprecong. 
This is more involved.
We make use of the following lemma.
\begin{lem}
Consider objects $\RelB$, $\RelBb$, $\SetV$, $\SetVv$, and morphisms
$p,q,p',q',f,g,f',g'$, making the following three diagrams 
commute.
\[
\xymatrix@R-20pt{
&\SetA\ar[dr]^f&
&
&\SetA\ar[r]^\CoalgA&\BehA\SetA\ar[dr]^{\BehA f}&&&\SetA\ar[dr]^{f'}
\\
\RelB\ar[ur]^p\ar[dr]_q\ar@{}[rr]|-{\text{(a)}}&&\SetV
&
\RelBb\ar[ur]^{p'}\ar[dr]_{q'}\ar@{}[rrr]|-{\text{(b)}}&&&\BehA\SetV
&
\RelB\ar[ur]^p\ar[dr]_q\ar@{}[rr]|-{\text{(c)}}&&\SetVv
\\
&\SetB\ar[ur]_g&
&
&\SetB\ar[r]_\CoalgB&\BehA\SetB\ar[ur]_{\BehA g}&&&\SetB\ar[ur]_{g'}
}\]
If the left-hand diagram is a pullback, and $\BehA$ preserves pullbacks
along monos, and $\CatA$ is regular, 
then the following diagram also commutes.
\[
\xymatrix@R-20pt{
&\SetA\ar[r]^\CoalgA&\BehA\SetA\ar[dr]^{\BehA f'}
\\
\RelBb\ar[ur]^{p'}\ar[dr]_{q'}\ar@{}[rrr]|-{\text{(d)}}&&&\BehA\SetVv
\\
&\SetB\ar[r]_\CoalgB&\BehA\SetB\ar[ur]_{\BehA g'}
}\]
\label{lemma:amcong}
\end{lem}
\begin{proofof}{\lemmaref{lemma:amcong}}
Write $\imMapA$ for the image of $\MapA:\SetA\to \SetV$, \etc.
Subdivide the pullback~(a) as follows.
\begin{equation}
\label{dgm:subdiv}
\xymatrix@R-24pt{
&&\SetA\ar@{-|>}[dr]
\\
&*+<8pt>{\imProjA}\pb{0}\ar@{ >->}[ur]
\ar@{-|>}[dr]
&&\imMapA\ar@{ >->}[dr]
\\
\RelB\ar@{-|>}[ur]\ar@{-|>}[dr]\pb{0}&&
*+<8pt>{W}\pb{180}\pb{0}
\ar@{ >->}[ur]
\ar@{ >->}[dr]
&& \SetV
\\
&*+<8pt>{\imProjB}\ar@{ >->}[dr]\pb{0}
\ar@{-|>}[ur]
&&\imMapB\ar@{ >->}[ur]
\\
&&\SetB\ar@{-|>}[ur]
}
\end{equation}
Since $\CatA$ is regular, 
the composite $\RelB\to W$ is regular epi. 
In this situation, 
the leftmost pullback is also a pushout, as indicated
(see e.g.~\cite[Thm~5.2]{ckp-maltsev}.)

Since diagram~(b) commutes, and since $\BehA$ preserves pullbacks along monos,
we have unique morphisms
${\RelBb\to\BehA(\imProjA)}$ and ${\RelBb\to \BehA(\imProjB)}$ making the following diagram commute. 
\[
\xymatrix@R-18pt{
  &\SetA\ar[rr]^\CoalgA&&\BehA\SetA\ar[dr]
  \\
  &&\BehA(\imProjA)\pbbb{0}\ar[dr]\ar[ur]&&\BehA(\imMapA)\ar[dr]
  \\
  \RelBb\ar[uur]^{\ProjAa}\ar[ddr]_{\ProjBb}
  \ar@{..>}[urr]\ar@{..>}[drr]&&&\BehA W\ar[ur]\ar[dr]
  \pb{0}&&\BehA\SetV
  \\
  &&\BehA(\imProjB)\pbbb{0}\ar[dr]\ar[ur]&&\BehA(\imMapB)\ar[ur]
  \\
  &\SetB\ar[rr]_\CoalgB&&\BehA\SetB\ar[ur]
}
\]
Now consider diagram~(c).
Since $W$ is a pushout, we have a unique morphism
$W\to\SetVv$ making the following diagram commute.
\[
\xymatrix@R-18pt{
&\imProjA\ar@{ >->}[r]\ar[dr]&\SetA\ar[dr]^{\MapAa}
\\
\RelB\ar@{-|>}[ur]\ar@{-|>}[dr]&&W\ar@{..>}[r]\pb{180}&\SetVv
\\
&\imProjB\ar@{ >->}[r]\ar[ur]&\SetB\ar[ur]_{\MapBb}}
\]
We can now conclude diagram~(d), by combining the 
previous two diagrams as follows.
\[
\xymatrix@R-18pt{
  &\SetA\ar[rr]^\CoalgA&&\BehA\SetA\ar[ddr]^{\BehA f'}
  \\
  &&\BehA(\imProjA)\ar[dr]\ar[ur]
  \\
  \RelBb\ar[uur]^{\ProjAa}\ar[ddr]_{\ProjBb}\ar[urr]\ar[drr]&&&
  \BehA W\ar[r]&\BehA\SetVv
  \\
  &&\BehA(\imProjB)\ar[dr]\ar[ur]
  \\
  &\SetB\ar[rr]_\CoalgB&&\BehA\SetB\ar[uur]_{\BehA g'}
}
\]
Thus \lemmaref{lemma:amcong} is proved.
\end{proofof}
Returning to the proof of \thmref{thm:relate}, 
notice that \itemref{relate:six} follows from \lemmaref{lemma:amcong},
in the case $S=S'=R$, $V=Z$. 

\Itemref{relate:seven} of \thmref{thm:relate} gives conditions under which 
every \HJbisimtext\ is an \AMbisimtext. 
The following fact is crucial here:
\begin{fact}
\label{fact:splitcover}
An \HJbisimtext\ $\RelA$ is an \AMbisimtext\
if the cover $\BehA\RelA\coverto\LiftedBehA\RelA$ is split. 
\qed
\end{fact}

Case (\ref{relate:sevena}), where all epis split, is thus trivial, 
and in case (\ref{relate:sevenb}), where $\BehA$ preserves relations, 
the cover is an isomorphism.
In cases (\ref{relate:sevenc}) and (\ref{relate:sevend}), we define a section 
${\LiftedBehA\RelA\to\BehA\RelA=\SPow\SmallS(\BehAa\RelA)}$ by defining 
the following composite relation
${(\LiftedBehA\RelA\leftarrow {\bullet}\to \BehAa\RelA)}$:
\[
\xymatrix@C-6pt@R-6pt{
&
\ar[dl]_=\LiftedBehA\RelA\ar[dr]
&&
{\ni_{\BehAa\SetA}}\times {\ni_{\BehAa\SetB}}
\ar[dr]\ar[dl]
&&
\BehAa\RelA
\ar[dl]
\ar[dr]^=
\\
\LiftedBehA\RelA
&&
\BehA\SetA\times\BehA\SetB
&&
\BehAa\SetA\times\BehAa\SetB
&&
\BehAa\RelA
}
\]
In case (\ref{relate:sevenc}), we must check that the composite relation is 
an \srelation\SmallS. By the axioms of open maps, 
it is sufficient to check that all 
the leftwards morphisms in the composite are in~$\SmallS$.
The right-most leftwards morphism is in~$\SmallS$ because~%
$\BehAa$ preserves relations, hence it is monic.

This concludes our proof of \thmref{thm:relate}.

\subsection{A note about equivalence relations}
\label{sec:equiv}
When $(\SetA,\CoalgA)=(\SetB,\CoalgB)$ then 
the maximal bisimulation, when it exists, is often 
an equivalence relation.
Some authors focus attention on those bisimulation relations that 
are equivalence relations. 
In this setting, it is reasonable to adjust the definition of
kernel bisimulation, so that the two coalgebra homomorphisms
$\SetA\to\SetZ$ are required to be equal.
An appropriate adjustment of \thmref{thm:relate} still holds
even when the pullback-preservation requirements are weakened as follows:
\begin{enumerate}[$\bullet$]
\item In \itemref{relate:three}, it is not necessary for $\CatA$ to have pushouts, 
  it is sufficient for $\CatA$ to have effective equivalence 
  relations (i.e., that every equivalence relation arises
  as a kernel pair);
\item In \itemref{relate:four}, 
  it is not necessary for $\BehA$ to preserve weak pullbacks,
  it is sufficient for $\BehA$ 
  to weakly preserve kernel pairs;
\item In \itemref{relate:five}, 
  it is not necessary for $\BehA$ to cover pullbacks,
  it is sufficient that $\BehA$ covers
  kernel pairs;
\item In \itemref{relate:six}, 
  it is not necessary for $\BehA$ to preserve pullbacks along monos,
  it is sufficient for $\BehA$ to 
  preserve monos. (In \dgmref{dgm:subdiv},
  if $f=g$ then ${W=\imMapA=\imMapB}$.)
\end{enumerate}
The conditions are connected: in an extensive regular category, a
functor covers pullbacks if and only if it covers kernel pairs and
preserves pullbacks along monos. (Gumm and Schr\"oder~\cite{gs-typecoalg} 
showed this for $\Set$, and their argument
is readily adapted to this more general setting.).

In summary, we have the following situation,
when focusing on equivalence relations:

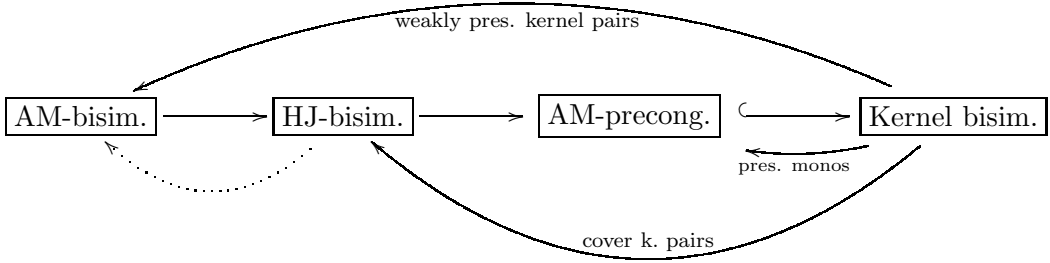
\[\xymatrix@R-40pt@C+15pt{
*+<5pt>{\txt{\framebox{AM-bisim.}}}
\ar[r]^{}
&
*+<5pt>{\txt{\framebox{HJ-bisim.}}}
\ar[r]^(.45){}
\ar@{..>}@/^1cm/[l]^{}
&
*+<12pt>{\txt{\framebox{AM-precong.}}\ }
\ar@{^{(}->}[r]^{}
&
*+<8pt>{\txt{\framebox{Kernel bisim.}}}
\ar@/_1.5cm/[lll]^{\text{weakly pres. kernel pairs}}
\ar@/^1.9cm/[ll]_{\text{cover k.~pairs}}
\ar@/^.5cm/[l]^{\txt{\tiny pres. monos}}
}\]

\section{Constructing bisimilarity}
\label{sec:termrelseq}

\noindent In this section we consider a procedure for constructing
the maximal bisimulation.
We relate it with the terminal sequence,
which is used for finding
final coalgebras.

\subsubsection*{Context.}
In this section we 
assume that the ambient category~$\CatA$ is complete and regular.
We 
fix an endofunctor~$\ArbBehA$ 
on~%
$\CatA$, and fix two \gcoalgtext\ArbBehA s,
${\CoalgA:\ArbCarrierA\to\ArbBehA\ArbCarrierA}$
and
${\CoalgB:\ArbCarrierB\to\ArbBehA\ArbCarrierB}$.

\subsection{The relation refinement sequence}

The greatest \HJbisimtext s can be understood 
as greatest fixed points of the operator~%
$\RelOP$ on~$\RPOrder$.
We define an ordinal indexed cochain 
${(\RelSeqMapA\ArbOrdB\ArbOrdA:
  \RelSeqHJObjA\ArbOrdB\monoto \RelSeqHJObjA\ArbOrdA)_{\ArbOrdA\leq\ArbOrdB}}$
in $\RPOrder$,
in the usual way: 
\begin{enumerate}[$\bullet$]
\item \emph{Limiting case}:
  If $\ArbLimOrdA$ is limiting, then 
  let $\RelSeqHJObjA{\ArbLimOrdA}$ be
  $\bigcap_{\ArbOrdA<\ArbLimOrdA}\RelSeqHJObjA\ArbOrdA$,
  i.e.~the
  limit of 
  the cochain 
  ${(\RelSeqMapA\ArbOrdB\ArbOrdA:
  \RelSeqHJObjA\ArbOrdB\monoto 
  \RelSeqHJObjA\ArbOrdA)_{\ArbOrdA\leq\ArbOrdB<\ArbLimOrdA}}$.
  In particular, let $\RelSeqHJObjA{0}=\SetA\times\SetB$.
\item \emph{Inductive case}:
  let $\RelSeqHJObjA{\ArbOrdA+1}=\RelOP(\RelSeqHJObjA\ArbOrdA)$.
\end{enumerate}
We call this cochain the \emph{relation refinement 
sequence}. 
If this sequence is eventually stationary
then it achieves the maximal post-fixed
point of~$\RelOP$, the greatest \HJbisimtext.
(NB. The sequence always converges when $\RPOrder$ is
small, e.g. when $\CatA$ is well-powered.)

For the case of the endofunctor
${\SetPow(\CCSLabels\times(-))}$ on~$\Set$,
the relation refinement sequence is a transfinite extension 
of Milner's sequence
${\sim_0}\pgeq{\sim_1}\pgeq\dots\pgeq{\sim}$
(e.g.~\cite[\textcitesec~5.7]{m-ccsA}),
studied from an algorithmic perspective 
by Kanellakis and Smolka~\cite{ks-partref}.

If $\CatA$ has pushouts, we can also consider a cochain 
$(\RelSeqAMObjA\ArbOrdB\monoto \RelSeqAMObjA\ArbOrdA)_{\ArbOrdA\leq\ArbOrdB}$
corresponding to the operator~$\RelCOP$. As will be seen, the 
two sequences of relations often coincide. 

(The other notions of bisimulation, \AMbisimtext\ and kernel bisimulation,
cannot be 
characterized as post-fixed points, in general.)

\subsection{The terminal sequence}
\newcommand{\ArbTermSeqMapA}[2]{z_{#1,#2}}
\newcommand{\ArbTermSeqObjA}[1]{Z_{#1}}

There is another sequence that
is often studied in the coalgebraic setting.
The \emph{terminal sequence}
is an ordinal-indexed cochain 
\[{(\ArbTermSeqMapA\ArbOrdB\ArbOrdA:
\ArbTermSeqObjA\ArbOrdB
\to 
\ArbTermSeqObjA\ArbOrdA)_{\ArbOrdA\leq\ArbOrdB}}\]
that can be used to construct
a final coalgebra for an endofunctor~%
(e.g.~Worrell~\cite{w-finalseq}). 
The 
idea is to begin with the 
terminal object, and then successively find 
\emph{\galgtext\ArbBehA} structures, so that if the sequence
converges, i.e.\ the \galgtext\ArbBehA\ structure
is an isomorphism, then we have a
\emph{\gcoalgtext\ArbBehA}\ structure,
and indeed the final such. The cochain commutes and satisfies the following 
conditions:
\begin{enumerate}[$\bullet$]
\item \emph{Limiting case:}
$\ArbTermSeqObjA{\ArbLimOrdA}
=
\lim\setofcomp{\ArbTermSeqMapA\ArbOrdB\ArbOrdA:\ArbTermSeqObjA\ArbOrdB\to\ArbTermSeqObjA\ArbOrdA}{\,{\ArbOrdA\leq\ArbOrdB<\ArbLimOrdA}}$, if
$\ArbLimOrdA$ is limiting,
in which case the cone 
$\setofcomp{\ArbTermSeqMapA\ArbLimOrdA\ArbOrdA:
  \ArbTermSeqObjA\ArbLimOrdA\to\ArbTermSeqObjA\ArbOrdA}{\,\ArbOrdA<\ArbLimOrdA}$
is the limiting one;
\item \emph{Inductive case:}
$\ArbTermSeqObjA{\ArbOrdA+1}=
\BehA(\ArbTermSeqObjA{\ArbOrdA})$;
and 
$\ArbTermSeqMapA{\ArbOrdB+1}{\ArbOrdA+1}=
\BehA(\ArbTermSeqMapA\ArbOrdB\ArbOrdA):
\ArbTermSeqObjA{\ArbOrdB+1}\to
\ArbTermSeqObjA{\ArbOrdA+1}
$.
\end{enumerate}
\newcommand{\ArbFinalMapA}[1]{x_{#1}}
\newcommand{\ArbFinalMapB}[1]{y_{#1}}
\subsection{Relating the relation and terminal sequences}
The coalgebras 
$(\ArbCarrierA,\CoalgA)$,
$(\ArbCarrierB,\CoalgB)$
determine two cones 
\[(\ArbFinalMapA\ArbOrdA:
\ArbCarrierA\to \ArbTermSeqObjA\ArbOrdA)_\ArbOrdA
\qquad
(\ArbFinalMapB\ArbOrdA:
\ArbCarrierB\to \ArbTermSeqObjA\ArbOrdA)_\ArbOrdA
\]
over the terminal sequence. The first cone,~%
$(\ArbFinalMapA\ArbOrdA)_\ArbOrdA$ is 
given as follows.
\begin{enumerate}[$\bullet$]
\item \emph{Limiting case:} 
  If $\ArbLimOrdA$ is a limit ordinal
  then the morphisms
  ${\ArbFinalMapA{\ArbOrdA}:
  \ArbCarrierA\to\ArbTermSeqObjA{\ArbOrdA}}$
  for ${\ArbOrdA<\ArbLimOrdA}$
  form a cocone over the cochain 
  ${(\ArbTermSeqMapA{\ArbOrdB}{\ArbOrdA}:
  \ArbTermSeqObjA{\ArbOrdB}\to\ArbTermSeqObjA{\ArbOrdA})
  _{\ArbOrdA\leq\ArbOrdB<\ArbLimOrdA}}$, with apex
  $\ArbCarrierA$.
  We let 
  ${\ArbFinalMapA{\ArbLimOrdA}:\ArbCarrierA\to
  \ArbTermSeqObjA\ArbLimOrdA}$ 
  be the unique mediating morphism.
  For instance,~when $\ArbLimOrdA=0$,
  then 
  ${\ArbFinalMapA\ArbLimOrdA:
  \ArbCarrierA\to\ArbTermSeqObjA\ArbLimOrdA}$ 
  is the terminal map
  ${\ArbCarrierA\to 1}$.
\item 
  \emph{Inductive case:} 
  Let 
  $\ArbFinalMapA{\ArbOrdA+1}$ be the composite
  \[\ArbCarrierA
    \stackrel\CoalgA\longrightarrow
    \ArbBehA\ArbCarrierA
    \stackrel{\ArbBehA\ArbFinalMapA{\ArbOrdA}}\longrightarrow
    \ArbBehA\ArbTermSeqObjA{\ArbOrdA}
    =
    \ArbTermSeqObjA{\ArbOrdA+1}
  \quad\eqnstop\]
\end{enumerate}
The other cone, $(\ArbFinalMapB\ArbOrdA:
\ArbCarrierA\to \ArbTermSeqObjA\ArbOrdA)_\ArbOrdA$,
is defined similarly.

These cones determine another ordinal indexed cochain
${(\RelSeqObjA\ArbOrdB\monoto \RelSeqObjA\ArbOrdA)_{\ArbOrdA\leq\ArbOrdB}}$
in $\RPOrder$.
For every ordinal $\ArbOrdA$, let $\RelSeqObjA\ArbOrdA$ be the 
following pullback.
\[
\xymatrix{
  {\RelSeqObjA\ArbOrdA}\pbb{-45}
  \ar[r]
  \ar[d]
  &
  {\ArbCarrierB}
  \ar[d]^{\ArbFinalMapB\ArbOrdA}
  \\
  {\ArbCarrierA}
  \ar[r]_{\ArbFinalMapA\ArbOrdA}
  &
  {\ArbTermSeqObjA\ArbOrdA}
}
\]

\begin{thm}
  \label{thm:reltermseq}
Consider an ordinal $\ArbOrdA$.
\begin{enumerate}[\em(1)]
\item 
$\RelSeqObjA\ArbOrdA$ contains all the kernel bisimulations
and all the \AMprecong s.
\item 
If $\ArbBehA$ covers pullbacks,
then $\RelSeqObjA\ArbOrdA=\RelSeqHJObjA\ArbOrdA$.
\item 
If $\ArbBehA$ preserves pullbacks along monos, and $\CatA$ has pushouts,
then ${\RelSeqObjA\ArbOrdA=\RelSeqAMObjA\ArbOrdA}$.
\end{enumerate}
\end{thm}
\newcommand{\ArbBehImageSeqObjA}[1]{I_{#1}}
\begin{proof}
To see that
$\RelSeqObjA\ArbOrdA$ contains all the kernel bisimulations,
notice that for a cospan of coalgebras,
$(\SetA\rightarrow\SetZ\leftarrow\SetB)$,
the coalgebra $\SetZ$ determines a cone over the terminal sequence,
just as $\SetA$ and $\SetB$ do.
The relation $\RelSeqObjA\ArbOrdA$ contains all the \AMprecong s by
transfinite induction on $\ArbOrdA$: the limit step is vacuous and 
the inductive step uses the definition of \AMprecong.
Statements (2) and (3) are also proved by transfinite induction
on $\ArbOrdA$:
the limit steps use the fact that
limits commute with limits;
the inductive steps 
follow from \factref{fact:coverpbacks},
and from \lemmaref{lemma:amcong}, respectively.
\end{proof}
\subsubsection*{Note}
A consequence of Item~(3) of \thmref{thm:reltermseq} is that,
if the sequence ${(\RelSeqObjA\ArbOrdB\monoto \RelSeqObjA\ArbOrdA)_{\ArbOrdA\leq\ArbOrdB}}$
converges, then the result is the greatest \AMprecong,
provided the endofunctor preserves pullbacks along monos 
and $\CatA$ has pushouts.
In fact, this corollary still holds even if $\CatA$ does not have all pushouts.
This can be proved directly by transfinite induction; the inductive
step uses \lemmaref{lemma:amcong}.

\subsection{Convergence for the relation refinement sequences}
By \thmref{thm:reltermseq}, if 
the terminal sequence converges, then 
the relation refinement sequence does too.
Of course, this is not a sufficient condition. 
Indeed, even when there is a final coalgebra, the relation refinement sequence
may converge before the terminal sequence:
\begin{prop}
  If, for some ordinal $\ArbOrdA$, the morphism
  ${\ArbTermSeqMapA{\ArbOrdA+1}\ArbOrdA:\ArbTermSeqObjA{\ArbOrdA+1}
  \to\ArbTermSeqObjA\ArbOrdA}$ is monic, 
  then $\RelSeqObjA{\ArbOrdA}=\RelSeqObjA{\ArbOrdA+1}$.
\end{prop}
\begin{proof}
  We have the following situation.
  \[
  \xymatrix@C+1.3cm{
    \RelSeqObjA{\ArbOrdA+1}\ar[r]\ar[d]\pb{-45}
    &\ArbTermSeqObjA{\ArbOrdA+1}\ar[d]^\Delta\ar[r]^{\ArbTermSeqMapA{\ArbOrdA+1}\ArbOrdA}
    &\ArbTermSeqObjA{\ArbOrdA}\ar[d]^\Delta
    \\
    \SetA\times\SetB
    \ar[r]_(.45){\ArbFinalMapA{\ArbOrdA+1}\times\ArbFinalMapB{\ArbOrdA+1}}
    &\ArbTermSeqObjA{\ArbOrdA+1}\times\ArbTermSeqObjA{\ArbOrdA+1}
    \ar[r]_(.55){\ArbTermSeqMapA{\ArbOrdA+1}\ArbOrdA\times\ArbTermSeqMapA{\ArbOrdA+1}\ArbOrdA}
    &
    \ArbTermSeqObjA{\ArbOrdA}\times\ArbTermSeqObjA{\ArbOrdA}
  }
  \]
  The left-hand square is a pullback --- this is a rearrangement of 
  the definition of $\RelSeqObjA{\ArbOrdA+1}$.
  The right-hand square is a pullback 
  if and only if $\ArbTermSeqMapA{\ArbOrdA+1}\ArbOrdA$ is monic.
  Thus the outer square is a pullback.
  Now 
  $\ArbFinalMapA{\ArbOrdA}=
  (\ArbFinalMapA{\ArbOrdA+1}\cdot \ArbTermSeqMapA{\ArbOrdA+1}{\ArbOrdA})$
  and 
  $\ArbFinalMapB{\ArbOrdA}=
  (\ArbFinalMapB{\ArbOrdA+1}\cdot \ArbTermSeqMapA{\ArbOrdA+1}{\ArbOrdA})$,
  which means that 
  $\RelSeqObjA{\ArbOrdA+1}=\RelSeqObjA\ArbOrdA$.
\end{proof}
If $\CatA$ is $\Set$ and $\BehA$ preserves filtered colimits, then the
terminal sequence does not converge until~$(\omega+\omega)$, but it
becomes monic at $\omega$ \cite{w-finalseq}.  As is well-known, the
relation refinement sequence for image-finite transition systems
converges at~$\omega$.

For the case when ${\CatA=\Set}$ and
sets~%
${\ArbCarrierA}$ and~${\ArbCarrierB}$ are both finite,
the relation refinement sequence
will converge before $\omega$
because the skeleton of 
$\RPOrder$ is finite.
This is relevant in a 
slightly more general setting:
In a Boolean Grothendieck topos,
every descending \gcochain\omega\ of subobjects
from a finitely presentable object is eventually constant.
\hide{
\newcommand{\ArbObjSeqObjA}[1]{S_{#1}}
\newcommand{\ArbObjSeqMapA}[2]{m_{#1,#2}}
\newcommand{\ArbObjSeqObjEGA}[2]{S}
\newcommand{\ArbObjSeqMapEGA}[2]{m}
  \newcommand{\NatA}{i}
  \newcommand{\NatB}{j}
\begin{thm}
\label{thm:relseqterm}
Let $\ArbLimOrdA$ be a limit ordinal 
for which $\card(\ArbLimOrdA)$ is regular.
Let $\CatA$ be a Boolean topos
with colimits of bounded \gchain\ArbLimOrdA s of monos.
Consider a \gcochain\ArbLimOrdA\ 
$(\ArbObjSeqMapA\ArbOrdB\ArbOrdA:
\ArbObjSeqObjA\ArbOrdB\monoto\ArbObjSeqObjA\ArbOrdA)_
{\ArbOrdA\leq\ArbOrdB<\ArbLimOrdA}$
of monos
such that~%
$\ArbObjSeqObjA0$ is \gpresentable{(\card{(\ArbLimOrdA)})}.
Then there is an ordinal $\ArbOrdA<\ArbLimOrdA$
for which~%
${\ArbObjSeqMapA{\ArbOrdA+1}\ArbOrdA:
\ArbObjSeqObjA{\ArbOrdA+1}\to\ArbObjSeqObjA{\ArbOrdA}}$ 
is an isomorphism.
\qed
\end{thm}}

\mbox{}

The presheaf topos $\PshfI$, used to model name-passing,
is not Boolean. It does, however, have 
$\ShfI$ as
a Boolean subcategory,
and this is perhaps a more appropriate universe
for name-passing calculi~%
(see e.g.~\cite{fs-namepassing}).
There, the finitely presentable objects are exactly those objects
that can be described by finite `named-sets with symmetry'\cite{fs-namepassing,gmm-sheafnamedsets}.
Thus the techniques of this article provide a general 
foundation for the coalgebra-inspired verification procedures of 
Ferrari, Montanari and Pistore~\cite{fmp-hda}.

\appendix
\section{Some concepts from categorical logic}

\noindent We recall some concepts from categorical logic: regular categories;
and powersets from algebraic set theory.
\subsection{Regular and extensive categories}
An \emph{image} of a morphism
$f:A\to C$ is a monomorphism
$m:B\monoto C$ 
through which $f$ factors, which is minimal in the sense
that, if $f$ factors through any other mono, $B'\monoto C$,
then~$B$ is a subobject of~$B'$.
In this setting, 
the factoring morphism $f:A\coverto B$ is called a \emph{cover};
its image is $B\xrightarrow{\mathrm{id}}B$.
A category \emph{has images} if every morphism has an image.

In a category with finite limits, covers are epimorphisms. They serve
as a generalization of `surjective function'. In this setting, 
other authors refer to covers as \emph{strong} epimorphisms.

A category with finite limits and images 
is said to be \emph{regular} if covers are stable under pullback,
i.e., if the following diagram is a pullback,
and if $f$ is a cover, then so is $f'$.
\begin{equation}
\label{dgm:pback}
\xymatrix{
A' \ar[r]^{g'}
\ar[d]_{f'}\pbb{-45}
&
A\ar[d]^{f}
\\
B'\ar[r]_{g}
&*+<8pt>{B}}
\end{equation}

Recall that a category is said to be \emph{extensive} if 
it has coproducts and they are disjoint and stable under 
pullback~(see e.g.~\cite{clw-extensive}).

\subsection{Open maps, powersets, and algebraic set theory}
We now recall an analysis of `smallness'
due to Joyal and Moerdijk~\cite{jm-open,jm-ast}.
(For a more recent introduction to 
this area of research, see the articles by Awodey~\cite{a-ast} and by 
van den Berg and Moerdijk~\cite{bm-introast}.)

An intuition for this analysis is that a morphism
$f\colon A\to B$ describes a \gindexed B\ family of classes --- 
informally, for each element $b$ of $B$, we have a class~${f\inv(b)}$. 
When we say that a morphism $f\colon A\to B$ is small,
an intuition is that each fibre ${f\inv(b)}$ is small.
In particular, we say that an object $A$ is small if the 
terminal map $A\to 1$ is small.

\subsubsection*{Open maps in regular categories.}
Let $\CatA$ be a regular category. 
A class $\SmallS$ of morphisms in~$\CatA$ is a class of \emph{open maps} 
if it satisfies the following four axioms. (This numbering follows
\cite{jm-open}.)
\begin{enumerate}[(1)]
\item[(A1)] $\SmallS$ is closed under composition, and all identity morphisms are in 
  $\SmallS$.
\item[(A2)] $\SmallS$ is stable under pullback, i.e., in \dgmref{dgm:pback},
  if $f$ is in $\SmallS$, then $f'$ is also in $\SmallS$.
\item[(A3)] (`Descent') In \dgmref{dgm:pback},
  if $f'$ is in $\SmallS$ and $g$ is a cover, then $f$ is also in~$\SmallS$.
\item[(A6)] In the following triangle, if $f$ is in $\SmallS$ and 
  $e$ is a cover, then~$g$ is also in~$\SmallS$.
  \[
  \xymatrix{
    A\ar@{-|>}[r]^e\ar[dr]_f&A'\ar[d]^g\\
    &B}
  \]
\end{enumerate}
Most authors assume that $\CatA$ has additional structure
so that the universal quantifier can be interpreted
in $\CatA$. We do not need that in this article.

\subsubsection*{Sums and open maps.}
In an extensive regular category, 
it is appropriate to assume the following additional axioms.
\begin{enumerate}[(1)]
\item[(A4)] The maps $0\to 1$ and $1+1\to 1$ are in $\SmallS$.
\item[(A5)] If $A\to A'$ and $B\to B'$ are in $\SmallS$,
then so is $(A+B)\to(A'+B')$.
\end{enumerate}

\subsubsection*{Powersets}
Given a regular category $\CatA$ and a class of open maps $\SmallS$,
an \srelation\SmallS\ is a jointly monic span $(I\leftarrow R \to A)$ 
whose left projection is in $\SmallS$. 

An \spowerset\SmallS\ for an object $A$ of $\CatA$ is an object
$\SPow\SmallS(A)$ together with an \srelation\SmallS\ 
${(\SPow\SmallS(A)\leftarrow{\ni_A}\to A)}$ 
such that for every 
\srelation\SmallS\ ${(I\leftarrow R \to A)}$ there is a unique 
morphism
$I\to \SPow\SmallS(A)$ making 
$(R\monoto I\times A)$ a pullback of $({\ni_A}\monoto \SPow\SmallS(A)\times A)$:
\[
\xymatrix{
R
\ar@{ >->}[d]\ar[r]&{\ni_A}
\ar@{ >->}[d]
\\
I \times A \ar[r]& \SPow\SmallS (A)\times A}
\]
It follows from axiom~(A2) that morphisms $I\to\SPow\SmallS (A)$ are in 
bijective correspondence with 
\srelation\SmallS s $(I\leftarrow R\to A)$.

If every object of $\CatA$ has an \spowerset\SmallS,
then the construction $\SPow\SmallS$ extends straightforwardly
to a covariant endofunctor on $\CatA$, as follows.
For any morphism $f:A\to B$,
the action ${\SPow\SmallS(f):\SPow\SmallS(A)\to\SPow\SmallS(B)}$ corresponds 
to the 
\srelation\SmallS\ in the image of the span
$(\SPow\SmallS(A)\leftarrow {\ni_A}\to A \to B)$,
using axiom~(A6).

\subsubsection*{Separation and collection}
There are various axioms and axiom schema that can be 
assumed as principles for defining sets in a constructive setting. The 
notes by Aczel and Rathjen~\cite{ar-cst} provide an overview.

The separation axiom (also known as bounded comprehension)
amounts to the following axiom on $\SmallS$.
\begin{enumerate}[(M)]
\item All monomorphisms in $\CatA$ are in $\SmallS$.
\end{enumerate}
The (strong) collection axiom has the following categorical counterpart,
when there is an \spowerset\SmallS:
\begin{enumerate}[(A9)]
\item The endofunctor 
  $\SPow\SmallS$ preserves covers.
\end{enumerate}

\subsubsection*{Further axioms.}
The axioms above are all that we need in this article.
As a foundation of mathematics, these axioms are too weak:
one would typically also require 
that there is a small natural numbers object; that
each powerset $\SPow\SmallS(X)$ is small;
and that there is a class of all sets, a universal small map.




\newcommand{\BIBprocof}[1]{Proceedings of #1}
\newcommand{\BIBtphols}[3]{\BIBprocof{the \OrdinalToName{#3} International Conference on Theorem Proving in Higher Order Logics (TPHOLs'#1)}}
\newcommand{\BIBsos}[3]{\BIBprocof{the \OrdinalToName{#3} Workshop on Structural Operational Semantics (SOS'#1)}}
\newcommand{\BIBconcur}[3]{\BIBprocof{CONCUR'#1}}
\newcommand{\BIBcalco}[3]{\BIBprocof{CALCO'#1}}
\newcommand{\BIBicalp}[3]{\BIBprocof{the \OrdinalToName{#3} International Colloquium on Automata, Languages and Programming (ICALP'#1)}}
\newcommand{\BIBlics}[3]{\BIBprocof{LICS'#1}}
\newcommand{\BIBcmcs}[3]{\BIBprocof{CMCS'#1}}
\newcommand{\BIBctcs}[3]{\BIBprocof{CTCS'#1}}
\newcommand{\BIBmfps}[3]{\BIBprocof{MFPS #3}}
\newcommand{\BIBpopl}[3]{\BIBprocof{the \OrdinalToName{#3} Annual ACM SIGPLAN -- SIGACT
 Symposium on Principles of Programming Languages 
(POPL'#1)}}
\newcommand{\BIBfossacs}[3]{\BIBprocof{FOSSACS'#1}}
\newcommand{\BIBieee}{IEEE Computer Society Press}
\newcommand{\BIBextabs}[1]{Extended abstract appeared in \emph{#1}}
\newcommand{\BIBprelim}[1]{Preliminary version appeared in \emph{#1}}
\newcommand{\abbrvstop}{.}

\bibliographystyle{abbrv}
\bibliography{refs.bib}

\newcommand{\SortNoop}[1]{}
\begin{thebibliography}{10}

\bibitem{a-cooksnonwellfounded}
S.~Abramsky.
\newblock A {C}ook's tour of the finitary non-well-founded sets.
\newblock In S.~N. Art{\"e}mov, H.~Barringer, A.~S. d'Avila Garcez, L.~C. Lamb,
  and J.~Woods, editors, {\em We Will Show Them! Essays in Honour of Dov
  Gabbay}, volume~1, pages 1--18. College Publications, 2005.

\bibitem{am-finalcoalg}
P.~Aczel and N.~Mendler.
\newblock A final coalgebra theorem.
\newblock In {\em \BIBctcs{89}{1989}{3}}, volume 389 of {\em Lecture Notes in
  Comput.\ Sci\abbrvstop}, pages 357--365. Springer, 1989.

\bibitem{ar-cst}
P.~Aczel and M.~Rathjen.
\newblock Notes on constructive set theory.
\newblock Technical Report~40, Institut Mittag-Leffler, Sweden, 2001.

\bibitem{a-introcoalg}
J.~Ad\'amek.
\newblock Introduction to coalgebra.
\newblock {\em Theory Appl.\ of Categ\abbrvstop}, 14(8):157--199, 2005.

\bibitem{a-ast}
S.~Awodey.
\newblock A brief introduction to algebraic set theory.
\newblock {\em Bull.~Symbolic Logic}, 14(3):281--298, 2008.

\bibitem{bsv-probtypes}
F.~Bartels, A.~Sokolova, and E.~P. de~Vink.
\newblock A hierarchy of probabilistic system types.
\newblock {\em Theoret.\ Comput.\ Sci\abbrvstop}, 327(1-2):3--22, 2004.

\bibitem{bm-nonwellfound}
B.~{\SortNoop{Berg}}van~den Berg and F.~De~Marchi.
\newblock Models of non-well-founded sets via an indexed final coalgebra
  theorem.
\newblock {\em J. Symbolic Logic}, 72(3):767--791, 2007.

\bibitem{bm-introast}
B.~{\SortNoop{Berg}}van~den Berg and I.~Moerdijk.
\newblock A unified approach to algebraic set theory.
\newblock In {\em \BIBprocof{Logic Colloquium 2006}}, volume~32 of {\em Lecture
  Notes in Logic}, 2009.

\bibitem{bfv-vietoris}
N.~Bezhanishvili, G.~Fontaine, and Y.~Venema.
\newblock Vietoris bisimulations.
\newblock {\em J. Logic Comput\abbrvstop}, 20(5):1017--1040, 2010.

\bibitem{bm-coalgreact}
F.~Bonchi and U.~Montanari.
\newblock Coalgebraic models for reactive systems.
\newblock In {\em \BIBconcur{07}{2007}{18}}, volume 4703 of {\em Lecture Notes
  in Comput.\ Sci\abbrvstop}, pages 364--379, 2007.

\bibitem{bk-pilogic}
M.~M. Bonsangue and A.~Kurz.
\newblock Pi-calculus in logical form.
\newblock In {\em \BIBlics{07}{07}{07}}, pages 303--312, 2007.

\bibitem{bhmw-pseudometric}
F.~{\SortNoop{Breugel}}van~Breugel, C.~Hermida, M.~Makkai, and J.~Worrell.
\newblock An accessible approach to behavioural pseudometrics.
\newblock In {\em \BIBprocof{ICALP'05}}, volume 3580 of {\em Lecture Notes in
  Comput.\ Sci\abbrvstop}, pages 1018--1030. Springer-Verlag, 2005.

\bibitem{ckp-maltsev}
A.~Carboni, G.~M. Kelly, and M.~C. Pedicchio.
\newblock Some remarks on {M}altsev and {G}oursat categories.
\newblock {\em Appl. Categ. Structures}, 1(4):385--421, 1993.

\bibitem{ckw-twocatgeom}
A.~Carboni, G.~M. Kelly, and R.~J. Wood.
\newblock A 2-categorical approach to geometric morphisms {I}.
\newblock {\em Cah.\ Topol.\ G\'eom.\ Diff\'er.\ Cat\'eg\abbrvstop},
  XXXII(1):47--95, 1991.

\bibitem{clw-extensive}
A.~Carboni, S.~Lack, and R.~F.~C. Walters.
\newblock Introduction to extensive and distributive categories.
\newblock {\em J.\ Pure Appl.\ Algebra}, 83:145--158, 1993.

\bibitem{ddlp-bisimcocong}
V.~Danos, J.~Desharnais, F.~Laviolette, and P.~Panangaden.
\newblock Bisimulation and cocongruence for probabilistic systems.
\newblock {\em Inform.\ and Comput\abbrvstop}, 204(4):503--523, 2006.

\bibitem{fmp-hda}
G.~Ferrari, U.~Montanari, and M.~Pistore.
\newblock Minimizing transition systems for name passing calculi: A
  co-algebraic formulation.
\newblock In {\em \BIBfossacs{02}{2002}{5}}, volume 2303 of {\em Lecture Notes
  in Comput.\ Sci\abbrvstop}, pages 129--158. Springer-Verlag, 2002.

\bibitem{fs-namepassing}
M.~P. Fiore and S.~Staton.
\newblock Comparing operational models of name-passing process calculi.
\newblock {\em Inform.\ and Comput\abbrvstop}, 204(4):435--678, 2006.

\bibitem{ft-namepassing}
M.~P. Fiore and D.~Turi.
\newblock Semantics of name and value passing (extended abstract).
\newblock In {\em \BIBlics{01}{2001}{16}}, pages 93--104, 2001.

\bibitem{f-numerology}
P.~Freyd.
\newblock Numerology in topoi.
\newblock {\em Theory Appl.\ of Categ\abbrvstop}, 16(19):522--528, 2006.

\bibitem{gmm-sheafnamedsets}
F.~Gadducci, M.~Miculan, and U.~Montanari.
\newblock About permutation algebras, (pre)sheaves and named sets.
\newblock {\em Higher-Order Symb. Comput\abbrvstop}, 19(2--3):283--304, 2006.

\bibitem{gyv-open}
N.~Ghani, K.~Yemane, and B.~Victor.
\newblock Relationally staged computations in calculi ofmobile processes.
\newblock In {\em \BIBcmcs{04}{2004}{7}}, volume 106 of {\em Electron.\ Notes
  Theor.\ Comput.\ Sci\abbrvstop}, pages 105--120, 2004.

\bibitem{g-coalgfunctorwpp}
H.~P. Gumm.
\newblock Functors for coalgebras.
\newblock {\em Algebra\ Universalis}, 45:135--147, 2001.

\bibitem{gs-typecoalg}
H.~P. Gumm and T.~Schr\"oder.
\newblock Types and coalgebraic structure.
\newblock {\em Algebra\ Universalis}, 53:229--252, 2005.

\bibitem{hj-indcoindfib}
C.~Hermida and B.~Jacobs.
\newblock Structural induction and coinduction in a fibrational setting.
\newblock {\em Inform.\ and Comput\abbrvstop}, 145(2):107--152, 1998.

\bibitem{jr-coalgtutorial}
B.~Jacobs and J.~Rutten.
\newblock A tutorial on coalgebras and coinduction.
\newblock {\em EATCS Bulletin}, 62:222--259, 1997.

\bibitem{jptww-catcoalg}
P.~Johnstone, J.~Power, T.~Tsujishita, H.~Watanabe, and J.~Worrell.
\newblock On the structure of categories of coalgebras.
\newblock {\em Theoret.\ Comput.\ Sci\abbrvstop}, 260(1-2):87--117, 2001.

\bibitem{jm-open}
A.~Joyal and I.~Moerdijk.
\newblock A completeness theorem for open maps.
\newblock {\em Ann.\ Pure\ Appl.\ Logic}, 70(1):51--86, 1994.

\bibitem{jm-ast}
A.~Joyal and I.~Moerdijk.
\newblock {\em Algebraic Set Theory}.
\newblock Cambridge University Press, 1995.

\bibitem{ks-partref}
P.~C. Kanellakis and S.~A. Smolka.
\newblock {C}{C}{S} expressions, finite state processes, and three problems of
  equivalence.
\newblock {\em Inform.\ and Comput\abbrvstop}, 86:43--68, 1990.

\bibitem{k-recbialg}
B.~Klin.
\newblock Adding recursive constructs to bialgebraic semantics.
\newblock {\em J.\ Log.\ Algebr.\ Program\abbrvstop}, 60--61:259--286, 2004.

\bibitem{kkv-stonecoalg}
C.~Kupke, A.~Kurz, and Y.~Venema.
\newblock Stone coalgebras.
\newblock {\em Theoret.\ Comput.\ Sci\abbrvstop}, 327(1-2):109--134, 2004.

\bibitem{kurz-thesis}
A.~Kurz.
\newblock {\em Logics for Coalgebras and Applications to Computer Science}.
\newblock PhD thesis, Ludwig-Maximilians-Universit\"at M\"unchen, 2000.

\bibitem{ls-probbisim}
K.~G. Larsen and A.~Skou.
\newblock Bisimulation through probabilistic testing (preliminary report).
\newblock In {\em Proceedings of POPL'89}, pages 344--352. ACM Press, 1989.

\bibitem{levy-sim}
P.~B. Levy.
\newblock Similarity quotients as final coalgebras.
\newblock Available from the author's webpage, October 2010.

\bibitem{le-topsettheory}
T.~Libert and O.~Esser.
\newblock On topological set theory.
\newblock {\em Math. Log. Q.}, 51(3):263--273, 2005.

\bibitem{m-fusion}
M.~Miculan.
\newblock A categorical model of the {F}usion calculus.
\newblock In {\em \BIBmfps{08}{2008}{XXIV}}, volume 218 of {\em Electron.\
  Notes Theor.\ Comput.\ Sci\abbrvstop}, pages 275--293, 2008.

\bibitem{m-ccsA}
R.~Milner.
\newblock {\em A calculus of communicating systems}, volume~92 of {\em Lecture
  Notes in Comput.\ Sci\abbrvstop}.
\newblock Springer-Verlag, 1980.

\bibitem{m-csa}
R.~Milner.
\newblock Calculi for synchrony and asynchrony.
\newblock {\em Theoret.\ Comput.\ Sci\abbrvstop}, 25:267--310, 1983.

\bibitem{m-ccs}
R.~Milner.
\newblock {\em Communication and Concurrency}.
\newblock Prentice Hall, 1989.

\bibitem{mpw-pi}
R.~Milner, J.~Parrow, and D.~Walker.
\newblock A calculus of mobile processes, {I} and {II}.
\newblock {\em Inform.\ and Comput\abbrvstop}, 100(1):1--77, 1992.

\bibitem{m-coalglogic}
L.~S. Moss.
\newblock Coalgebraic logic.
\newblock {\em Ann.\ Pure\ Appl.\ Logic}, 96:277--317, 1999.

\bibitem{p-bisim}
D.~Park.
\newblock Concurrency and automata on infinite sequences.
\newblock In {\em \BIBprocof{Theoretical Computer Science: 5th GI Conference}},
  volume 104 of {\em Lecture Notes in Comput.\ Sci\abbrvstop}, pages 167--183.
  Springer, 1981.

\bibitem{p-bialgrec}
G.~D. Plotkin.
\newblock Bialgebraic semantics and recursion (extended abstract).
\newblock In {\em \BIBcmcs{01}{2001}{4}}, volume 44(1) of {\em Electron.\ Notes
  Theor.\ Comput.\ Sci\abbrvstop}, pages 1--4, 2001.

\bibitem{r-relators}
J.~J. M.~M. Rutten.
\newblock Relators and metric bisimulations.
\newblock In {\em \BIBcmcs{98}{1998}{1}}, volume~11 of {\em Electron.\ Notes
  Theor.\ Comput.\ Sci\abbrvstop}, 1998.

\bibitem{r-univcoalg}
J.~J. M.~M. Rutten.
\newblock Universal coalgebra: a theory of systems.
\newblock {\em Theoret.\ Comput.\ Sci\abbrvstop}, 249(1):3--80, 2000.

\bibitem{s-openbisim}
D.~Sangiorgi.
\newblock A theory of bisimulation for the pi-calculus.
\newblock {\em Acta\ Inform\abbrvstop}, 33(1):69--97, 1996.

\bibitem{staton-thesis}
S.~Staton.
\newblock \emph{Name-passing process calculi: operational models and structural
  operational semantics}.
\newblock PhD thesis. Technical Report UCAM-CL-TR-688, University of Cambridge,
  Computer Laboratory, 2007.

\bibitem{s-lics08}
S.~Staton.
\newblock General structural operational semantics through categorical logic.
\newblock In {\em \BIBlics{08}{2008}{23}}, pages 166-- 177, 2008.

\bibitem{staton-calco09}
S.~Staton.
\newblock Relating coalgebraic notions of bisimulation.
\newblock In {\em \BIBcalco{09}{2009}{3}}, volume 5728 of {\em Lecture Notes in
  Comput.\ Sci\abbrvstop}, pages 191--205, 2009.

\bibitem{t-relautomata}
V.~Trnkov\'a.
\newblock General theory of relational automata.
\newblock {\em Fundam. Inform.}, 3(2), 1980.

\bibitem{v-finalmeas}
I.~D. Viglizzo.
\newblock Final sequences and final coalgebras for measurable spaces.
\newblock In {\em CALCO'05}, volume 3629 of {\em Lecture Notes in Comput.\
  Sci\abbrvstop}, pages 395--407. Springer, 2005.

\bibitem{vr-bisimprob}
E.~P. {\SortNoop{Vink}}de~Vink and J.~J. M.~M. Rutten.
\newblock Bisimulation for probabilistic transition systems: a coalgebraic
  approach.
\newblock {\em Theoret.\ Comput.\ Sci\abbrvstop}, 221:271--293, 1999.

\bibitem{w-omegacat}
J.~Worrell.
\newblock Coinduction for recursive data types: partial orders, metric spaces
  and omega-categories.
\newblock In {\em \BIBcmcs{00}{00}{00}}, volume~33 of {\em Electron.\ Notes
  Theor.\ Comput.\ Sci\abbrvstop}, 2000.

\bibitem{w-finalseq}
J.~Worrell.
\newblock On the final sequence of a finitary set functor.
\newblock {\em Theoret.\ Comput.\ Sci\abbrvstop}, 338:184--199, 2005.

\end{thebibliography}


\end{document}